\documentclass[12pt, doublespacing, onecolumn]{IEEEtran}

\makeatletter

\newcommand{\Rmnum}[1]{\expandafter\@slowromancap\romannumeral #1@}

\newcommand\abs[1]{\left\lvert #1 \right\rvert}
\makeatother

\usepackage{amsmath,bm}
\usepackage{amssymb}
\usepackage{extarrows}
\usepackage{algorithm}
\usepackage{algorithmic}
\usepackage{graphicx} 
\usepackage{epstopdf}
\usepackage{amsthm}
\usepackage{enumerate}
\usepackage[T1]{fontenc}
\usepackage{cite,bm,color,url,textcomp}
\usepackage[hidelinks,breaklinks=true]{hyperref}
\usepackage{svg}
\usepackage{accents}
\usepackage{stfloats}
\usepackage{lipsum}
\usepackage{tikz}
\usepackage{makecell}
\usepackage[flushleft]{threeparttable}
\newcommand\blfootnote[1]{%
  \begingroup
  \renewcommand\thefootnote{}\footnote{#1}%
  \addtocounter{footnote}{-1}%
  \endgroup
}
\usepackage{numprint}
\npthousandsep{,}

\usepackage{setspace,xspace}

\theoremstyle{plain}
\newtheorem{lemma}{Lemma}
\newtheorem{remark}{Remark}
\makeatletter
\def\thmhead@plain#1#2#3{%
  \thmname{#1}\thmnumber{\@ifnotempty{#1}{ }\@upn{#2}}%
  \thmnote{ {\the\thm@notefont#3}}}
\let\thmhead\thmhead@plain
\makeatother
\begin{document}
\newcommand{\fs}{\hspace{0.07in}}
\newcommand{\bs}{\hspace{-0.1in}}
\newcommand{\re}{{\rm Re} \, }
\newcommand{\e}{{\rm E} \, }
\newcommand{\p}{{\rm P} \, }
\newcommand{\cn}{{\cal CN} \, }
\newcommand{\n}{{\cal N} \, }
\newcommand{\ba}{\begin{array}}
\newcommand{\ea}{\end{array}}
\newcommand{\be}{\begin{displaymath}}
\newcommand{\ee}{\end{displaymath}}
\newcommand{\ben}{\begin{equation}}  
\newcommand{\een}{\end{equation}}
\newcommand{\bea}{\begin{equation}\begin{aligned}}
\newcommand{\eea}{\end{aligned}\end{equation}}      
\newcommand{\bena}{\begin{eqnarray}}
\newcommand{\eena}{\end{eqnarray}}
\newcommand{\beqa}{\begin{eqnarray*}}
\newcommand{\enqa}{\end{eqnarray*}}
\newcommand{\f}{\frac}
\newcommand{\bc}{\begin{center}}
\newcommand{\ec}{\end{center}}
\newcommand{\bi}{\begin{itemize}}
\newcommand{\ei}{\end{itemize}}
\newcommand{\benu}{\begin{enumerate}}
\newcommand{\eenu}{\end{enumerate}}
\newcommand{\bdes}{\begin{description}}
\newcommand{\edes}{\end{description}}
\newcommand{\bt}{\begin{tabular}}
\newcommand{\et}{\end{tabular}}
\newcommand{\vs}{\vspace}
\newcommand{\hs}{\hspace}
\newcommand{\sort}{\rm sort \,}

\newcommand \thetabf{{\mbox{\boldmath$\theta$\unboldmath}}}
\newcommand{\Phibf}{\mbox{${\bf \Phi}$}}
\newcommand{\Psibf}{\mbox{${\bf \Psi}$}}
\newcommand \alphabf{\mbox{\boldmath$\alpha$\unboldmath}}
\newcommand \betabf{\mbox{\boldmath$\beta$\unboldmath}}
\newcommand \gammabf{\mbox{\boldmath$\gamma$\unboldmath}}
\newcommand \deltabf{\mbox{\boldmath$\delta$\unboldmath}}
\newcommand \epsilonbf{\mbox{\boldmath$\epsilon$\unboldmath}}
\newcommand \zetabf{\mbox{\boldmath$\zeta$\unboldmath}}
\newcommand \etabf{\mbox{\boldmath$\eta$\unboldmath}}
\newcommand \iotabf{\mbox{\boldmath$\iota$\unboldmath}}
\newcommand \kappabf{\mbox{\boldmath$\kappa$\unboldmath}}
\newcommand \lambdabf{\mbox{\boldmath$\lambda$\unboldmath}}
\newcommand \mubf{\mbox{\boldmath$\mu$\unboldmath}}
\newcommand \nubf{\mbox{\boldmath$\nu$\unboldmath}}
\newcommand \xibf{\mbox{\boldmath$\xi$\unboldmath}}
\newcommand \pibf{\mbox{\boldmath$\pi$\unboldmath}}
\newcommand \rhobf{\mbox{\boldmath$\rho$\unboldmath}}
\newcommand \sigmabf{\mbox{\boldmath$\sigma$\unboldmath}}
\newcommand \taubf{\mbox{\boldmath$\tau$\unboldmath}}
\newcommand \upsilonbf{\mbox{\boldmath$\upsilon$\unboldmath}}
\newcommand \phibf{\mbox{\boldmath$\phi$\unboldmath}}
\newcommand \varphibf{\mbox{\boldmath$\varphi$\unboldmath}}
\newcommand \chibf{\mbox{\boldmath$\chi$\unboldmath}}
\newcommand \psibf{\mbox{\boldmath$\psi$\unboldmath}}
\newcommand \omegabf{\mbox{\boldmath$\omega$\unboldmath}}
\newcommand \Sigmabf{\hbox{$\bf \Sigma$}}
\newcommand \Upsilonbf{\hbox{$\bf \Upsilon$}}
\newcommand \Omegabf{\hbox{$\bf \Omega$}}
\newcommand \Deltabf{\hbox{$\bf \Delta$}}
\newcommand \Gammabf{\hbox{$\bf \Gamma$}}
\newcommand \Thetabf{\hbox{$\bf \Theta$}}
\newcommand \Lambdabf{\hbox{$\bf \Lambda$}}
\newcommand \Xibf{\hbox{\bf$\Xi$}}
\newcommand \Pibf{\hbox{\bf$\Pi$}}
\newcommand \abf{{\bf a}}
\newcommand \bbf{{\bf b}}
\newcommand \cbf{{\bf c}}
\newcommand \dbf{{\bf d}}
\newcommand \ebf{{\bf e}}
\newcommand \fbf{{\bf f}}
\newcommand \gbf{{\bf g}}
\newcommand \hbf{{\bf h}}
\newcommand \ibf{{\bf i}}
\newcommand \jbf{{\bf j}}
\newcommand \kbf{{\bf k}}
\newcommand \lbf{{\bf l}}
\newcommand \mbf{{\bf m}}
\newcommand \nbf{{\bf n}}
\newcommand \obf{{\bf o}}
\newcommand \pbf{{\bf p}}
\newcommand \qbf{{\bf q}}
\newcommand \rbf{{\bf r}}
\newcommand \sbf{{\bf s}}
\newcommand \tbf{{\bf t}}
\newcommand \ubf{{\bf u}}
\newcommand \vbf{{\bf v}}
\newcommand \wbf{{\bf w}}
\newcommand \xbf{{\bf x}}
\newcommand \ybf{{\bf y}}
\newcommand \zbf{{\bf z}}
\newcommand \rbfa{{\bf r}}
\newcommand \xbfa{{\bf x}}
\newcommand \ybfa{{\bf y}}
\newcommand \Abf{{\bf A}}
\newcommand \Bbf{{\bf B}}
\newcommand \Cbf{{\bf C}}
\newcommand \Dbf{{\bf D}}
\newcommand \Ebf{{\bf E}}
\newcommand \Fbf{{\bf F}}
\newcommand \Gbf{{\bf G}}
\newcommand \Hbf{{\bf H}}
\newcommand \Ibf{{\bf I}}
\newcommand \Jbf{{\bf J}}
\newcommand \Kbf{{\bf K}}
\newcommand \Lbf{{\bf L}}
\newcommand \Mbf{{\bf M}}
\newcommand \Nbf{{\bf N}}
\newcommand \Obf{{\bf O}}
\newcommand \Pbf{{\bf P}}
\newcommand \Qbf{{\bf Q}}
\newcommand \Rbf{{\bf R}}
\newcommand \Sbf{{\bf S}}
\newcommand \Tbf{{\bf T}}
\newcommand \Ubf{{\bf U}}
\newcommand \Vbf{{\bf V}}
\newcommand \Wbf{{\bf W}}
\newcommand \Xbf{{\bf X}}
\newcommand \Ybf{{\bf Y}}
\newcommand \Zbf{{\bf Z}}
\newcommand \Omegabbf{{\bf \Omega}}
\newcommand \Rssbf{{\bf R_{ss}}}
\newcommand \Ryybf{{\bf R_{yy}}}
\newcommand \Cset{{\cal C}}
\newcommand \Rset{{\cal R}}
\newcommand \Zset{{\cal Z}}
\newcommand{\otheta}{\stackrel{\circ}{\theta}}
\newcommand{\defeq}{\stackrel{\bigtriangleup}{=}}
\newcommand{\oabf}{{\bf \breve{a}}}
\newcommand{\odbf}{{\bf \breve{d}}}
\newcommand{\oDbf}{{\bf \breve{D}}}
\newcommand{\oAbf}{{\bf \breve{A}}}
\renewcommand \vec{{\mbox{vec}}}
\newcommand{\Acalbf}{\bf {\cal A}}
\newcommand{\calZbf}{\mbox{\boldmath $\cal Z$}}
\newcommand{\feop}{\hfill \rule{2mm}{2mm} \\}
\newtheorem{theorem}{Theorem}

\newcommand{\Rnum}{{\mathbb R}}
\newcommand{\Cnum}{{\mathbb C}}
\newcommand{\Znum}{{\mathbb Z}}
\newcommand{\Enum}{{\mathbb E}}
\newcommand{\Mnum}{{\mathbb M}}
\newcommand{\Nnum}{{\mathbb N}}
\newcommand{\Inum}{{\mathbb I}}

\newcommand{\Acal}{{\cal A}}
\newcommand{\Bcal}{{\cal B}}
\newcommand{\Ccal}{{\cal C}}
\newcommand{\Dcal}{{\cal D}}
\newcommand{\Ecal}{{\cal E}}
\newcommand{\Fcal}{{\cal F}}
\newcommand{\Gcal}{{\cal G}}
\newcommand{\Hcal}{{\cal H}}
\newcommand{\Ical}{{\cal I}}
\newcommand{\Ocal}{{\cal O}}
\newcommand{\Rcal}{{\cal R}}
\newcommand{\Zcal}{{\cal Z}}
\newcommand{\Xcal}{{\cal X}}
\newcommand{\zzbf}{{\bf 0}}
\newcommand{\zebf}{{\bf 0}}

\newcommand{\eop}{\hfill $\Box$}

\newcommand{\gss}{\mathop{}\limits}
\newcommand{\gs}{\mathop{\gss_<^>}\limits}

\newcommand{\circlambda}{\mbox{$\Lambda$
             \kern-.85em\raise1.5ex
             \hbox{$\scriptstyle{\circ}$}}\,}

\newcommand{\tr}{\mathop{\rm tr}}
\newcommand{\var}{\mathop{\rm var}}
\newcommand{\cov}{\mathop{\rm cov}}
\newcommand{\diag}{\mathop{\rm diag}}
\def\rank{\mathop{\rm rank}\nolimits}
\newcommand{\ra}{\rightarrow}
\newcommand{\ul}{\underline}
\def\Pr{\mathop{\rm Pr}}
\def\Re{\mathop{\rm Re}}
\def\Im{\mathop{\rm Im}}

\def\submbox#1{_{\mbox{\footnotesize #1}}}
\def\supmbox#1{^{\mbox{\footnotesize #1}}}

%
\newtheorem{Theorem}{Theorem}[section]
\newtheorem{Definition}{Definition}
\newtheorem{Proposition}{Proposition}
\newtheorem{Lemma}{Lemma}
\newtheorem{Corollary}{Corollary}
\newtheorem{Conjecture}[Theorem]{Conjecture}
\newtheorem{Property}{Property}

%
\newcommand{\ThmRef}[1]{\ref{thm:#1}}
\newcommand{\ThmLabel}[1]{\label{thm:#1}}
\newcommand{\DefRef}[1]{\ref{def:#1}}
\newcommand{\DefLabel}[1]{\label{def:#1}}
\newcommand{\PropRef}[1]{\ref{prop:#1}}
\newcommand{\PropLabel}[1]{\label{prop:#1}}
\newcommand{\LemRef}[1]{\ref{lem:#1}}
\newcommand{\LemLabel}[1]{\label{lem:#1}}
%

\newcommand \bbs{{\boldsymbol b}}
\newcommand \cbs{{\boldsymbol c}}
\newcommand \dbs{{\boldsymbol d}}
\newcommand \ebs{{\boldsymbol e}}
\newcommand \fbs{{\boldsymbol f}}
\newcommand \gbs{{\boldsymbol g}}
\newcommand \hbs{{\boldsymbol h}}
\newcommand \ibs{{\boldsymbol i}}
\newcommand \jbs{{\boldsymbol j}}
\newcommand \kbs{{\boldsymbol k}}
\newcommand \lbs{{\boldsymbol l}}
\newcommand \mbs{{\boldsymbol m}}
\newcommand \nbs{{\boldsymbol n}}
\newcommand \obs{{\boldsymbol o}}
\newcommand \pbs{{\boldsymbol p}}
\newcommand \qbs{{\boldsymbol q}}
\newcommand \rbs{{\boldsymbol r}}
\newcommand \sbs{{\boldsymbol s}}
\newcommand \tbs{{\boldsymbol t}}
\newcommand \ubs{{\boldsymbol u}}
\newcommand \vbs{{\boldsymbol v}}
\newcommand \wbs{{\boldsymbol w}}
\newcommand \xbs{{\boldsymbol x}}
\newcommand \ybs{{\boldsymbol y}}
\newcommand \zbs{{\boldsymbol z}}

\newcommand \Bbs{{\boldsymbol B}}
\newcommand \Cbs{{\boldsymbol C}}
\newcommand \Dbs{{\boldsymbol D}}
\newcommand \Ebs{{\boldsymbol E}}
\newcommand \Fbs{{\boldsymbol F}}
\newcommand \Gbs{{\boldsymbol G}}
\newcommand \Hbs{{\boldsymbol H}}
\newcommand \Ibs{{\boldsymbol I}}
\newcommand \Jbs{{\boldsymbol J}}
\newcommand \Kbs{{\boldsymbol K}}
\newcommand \Lbs{{\boldsymbol L}}
\newcommand \Mbs{{\boldsymbol M}}
\newcommand \Nbs{{\boldsymbol N}}
\newcommand \Obs{{\boldsymbol O}}
\newcommand \Pbs{{\boldsymbol P}}
\newcommand \Qbs{{\boldsymbol Q}}
\newcommand \Rbs{{\boldsymbol R}}
\newcommand \Sbs{{\boldsymbol S}}
\newcommand \Tbs{{\boldsymbol T}}
\newcommand \Ubs{{\boldsymbol U}}
\newcommand \Vbs{{\boldsymbol V}}
\newcommand \Wbs{{\boldsymbol W}}
\newcommand \Xbs{{\boldsymbol X}}
\newcommand \Ybs{{\boldsymbol Y}}
\newcommand \Zbs{{\boldsymbol Z}}

\newcommand \Absolute[1]{\left\lvert #1 \right\rvert}

\title{Golay Complementary Sequences of Arbitrary Length and Asymptotic Existence of Hadamard Matrices}

\author{Cheng Du, \textit{Student Member, IEEE,} \  Yi Jiang, \textit{Member, IEEE,}    
}
\maketitle
\blfootnote{The work was supported by National Natural Science Foundation of China Grant No. 61771005. ({\em Corresponding author: Yi Jiang}) 

C. Du and Y. Jiang are Key Laboratory for Information Science of Electromagnetic Waves (MoE), Department of Communication Science and Engineering, School of Information Science and Technology, Fudan University, Shanghai, China (E-mails: cdu15@fudan.edu.cn, yijiang@fudan.edu.cn).}
\begin{abstract}
In this work, we construct $4$-phase Golay complementary sequence (GCS) set of cardinality $2^{3+\lceil \log_2 r \rceil}$ with arbitrary sequence length $n$, where the $10^{13}$-base expansion of $n$ has $r$ nonzero digits. Specifically,  the GCS octets (eight sequences) cover all the lengths no greater than $10^{13}$. Besides, based on the representation theory of signed symmetric group, we construct Hadamard matrices from some special GCS to improve their asymptotic existence: there exist Hadamard matrices of order $2^t m$ for any odd number $m$, where $t = 6\lfloor \frac{1}{40}\log_{2}m\rfloor + 10$.
\end{abstract}


%
\IEEEpeerreviewmaketitle

\section{Introduction}
The Golay complementary sequence (GCS) set is a set of $L$ sequences whose respective aperiodic autocorrelations add to be a scaled $\delta$-function \cite{tseng1972complementary}. This time-domain property has been utilized in radar ranging \cite{pezeshki2008doppler} and channel estimation \cite{spasojevic2001complementary}. In the frequency domain, their power spectrums add to be flat everywhere, which is useful for reducing the peak-to-average power ratio (PAPR) of OFDM system \cite{davis1999peak}\cite{pai2022two} and the omnidirectional precoding in massive MIMO scenarios \cite{li2021construction}\cite{girnyk2021efficient}. 

An important research direction of GCS for engineering purposes is to improve the answers to the following two questions: given $L$ the cadinality of GCS set, what lengths can be covered? for covering arbitrary length, how large should the cadinality be? For example, flexible lengths would accommodate for flexible antenna array sizes in omnidirectional beamforming \cite{li2021construction} and flexible numbers of subcarriers in OFDM \cite{chen2017novel}, and small cardinality facilitates the adoption of the orthogonal space time block code in omnidirectional beamforming \cite{li2021construction} and the PAPR reduction of OFDM signals \cite{davis1999peak}.

GCS pairs may be the most important family of GCS. For 2-phase GCS pairs with entries $\{1, -1\}$ \cite{golay1961complementary}, the known lengths (referred to as $2$-phase Golay numbers) are of form $2^a10^b26^c$ where $a, b, c\in {\mathbb Z^{+}} \cup \{ 0\}$ \cite{turyn1974hadamard, borwein2004complete}; for 4-phase GCS pairs with entries $\{1, -1, i, -i\}$, where $i$ is a primitive $4$-th root, the known lengths (referred to as $4$-phase Golay numbers) are of form $2^{a+u}3^{b}5^{c}11^{d}13^{e}$ where $a, b, c, d, e, u\in {\mathbb Z^{+}} \cup \{ 0\}$, $b+c+d+e \leq a+2u+1$, $u \leq c+e$ \cite{craigen2002complex}. These lengths are exponentially sparse. 

The rareness of the GCS pairs is due to their multiplicative construction: the lengths are the products of some integers. In comparison, \cite{craigen2002complex} proposed an additive construction of $4$-phase GCS quad (i.e., $L=4$): the lengths can be the summation of two $4$-phase Golay numbers to cover $827$ integers no greater than $1000$ \cite{craigen2002complex}. The additive construction was rediscovered in \cite{wang2017method, wang2019generalized}. Furthermore, it was generalized to a multiplicative and additive method in \cite{du2023polyphase} to cover more lengths of GCS quads. But none of these recursive constructions can cover arbitrary length to the best knowledge of us. 

Interestingly, some direct constructions can produce GCS (and its variants) of arbitrary length \cite{chen2017novel, ghosh2022direct}. Using the generalized Boolean function, \cite{chen2017novel} constructed a GCS set of cardinality $2^{k+1}$ to cover arbitrary length $n = 2^{m-1} + \sum_{\alpha = 1}^{k-1} d_\alpha 2^{\pi \left(m-k+\alpha\right) -1} + d_0 2^v$ where $k<m,\ 0\leq v \leq m-k,\ d_a \in \{0, 1\}$ and $\pi$ is a permutation of $\{1, 2, \cdots, m\}$ satisfying some constraints. \cite{ghosh2022direct} constructed the complete complementary codes (CCC) from the multivariable function, which can degenerate into a GCS with sequence length $n = \prod_{i=1}^\lambda n_i^{m_i}$ where $n_i$ is prime and $m_i\geq 0$, and the cardinality is $\prod_{i=1}^\lambda n_i$. 

Another motivation for studying GCS comes from the construction of Hadamard matrices. The $2$-phase GCS quads with sequence length $n$ can be used to construct Hadamard matrices of order $4n$ \cite{goethals1970skew}, and if the length of these sequences can be arbitrary, then the celebrated Hadamard conjecture would be proved \cite{seberry2020hadamard}. Besides, the $4$-phase GCS pairs and some special GCS quads can be utilized to improve the asymptotic existence of Hadamard matrices \cite{craigen1995signed, craigen1997asymptotic, livinskyi2012asymptotic}, which states that there exist Hadamard matrices of order $2^{t} m$ for any odd number $m$ and $t$ increases logarithmically with respect to $m$ \cite{seberry2020hadamard}. The first asymptotic result of $t=\lfloor 2\log_2(m-3)\rfloor +1, m>3$ was established in 1976 \cite{wallis1976existence}, then $t=4\lfloor \frac{1}{6}\log_2\frac{m-1}{2}\rfloor +6$ in 1995 \cite{craigen1995signed}, $t=4\lfloor \frac{1}{10}\log_2(m-1)\rfloor +6$ in 1997 \cite{craigen1997asymptotic}, and $t=6\lfloor \frac{1}{26}\log_2\frac{m-1}{2}\rfloor +11$ in 2012 \cite{livinskyi2012asymptotic}. \footnote{Another result of $t=6\lfloor \frac{1}{30}\log_2\frac{m-1}{2}\rfloor +13$ in \cite{livinskyi2012asymptotic} is problematic since it relies on \cite[Theorem 4]{koukouvinos1991addendum}, which is incorrect.} 

In this paper, we focus on the construction of $4$-phase GCS sets and Hadamard matrices. Our contributions are two-fold. First, we propose a multiplicative and additive construction of GCS set, which can cover arbitrary lengths with much smaller cardinality than the direct constructions. Specifically, the $4$-phase GCS octet (i.e., $L = 8$) covers all the lengths no greater than $10^{13}$. Second, by using the matrix representation of signed symmetric group \cite{craigen1995signed}, we propose a new construction of Hadamard matrices from $4$-phase GCS pairs and some special $4$-phase GCS quads, which leads to an improved asymptotic existence of Hadamard matrices: there exist Hadamard matrices of order $2^t m$ for any odd number $m$, where $t = 6\lfloor \frac{1}{40}\log_{2}m\rfloor + 10$.



Notations: for a sequence, $(\cdot)^*$ represents flipping and conjugating the sequence, $\underline{(\cdot)}$ and $\overline{(\cdot)}$ represent negating and conjugating the sequence, respectively; for a complex value, $\overline{(\cdot)}$ represents the complex conjugate. $\otimes$ represents the Kronecker product, $\vert$ represents concatenating two sequences, and $\otimes$ has higher precedence than $\vert$. For two sets of integers $\mathcal{S}$ and $\mathcal{T}$, $\mathcal{S} + \mathcal{T} \triangleq \{s+t\,\vert\,s\in \mathcal{S}, t\in \mathcal{T}\}$, $\mathcal{S} \cdot \mathcal{T} \triangleq \{st\,\vert\,s\in \mathcal{S}, t\in \mathcal{T}\}$ and $\mathcal{S}^2 = \mathcal{S} \cdot \mathcal{S}$. The number of elements in a set $\mathcal{S}$ is denoted by $\abs{\mathcal{S}}$. $\lfloor \cdot \rfloor$ represents taking the floor of a number. $\Znum_n \triangleq \{0, 1, \cdots, n-1\}$. $\Cnum$ is the field of complex numbers. $\text{lcm}(a, b, \cdots)$ stands for the least common multiple of integers $a, b, c,\dots$.

\section{Preliminaries} \label{SEC:pre}
In this section, we first introduce the background of the GCS set in Section \ref{sec:pre.gcs}, which is prepared for the construction of GCS set of arbitrary length in Section \ref{SEC:gcs}; we then explain the concept of perfect sequences over a signed symmetric group in Section \ref{sec:pre.spv}, which is helpful for improving the asymptotic existence of Hadamard matrices in Section \ref{sec:hadamard}.

\subsection{Golay Complementary Sequence Set} \label{sec:pre.gcs}
For two complex-valued sequences $\abf = [a_0, a_1, \cdots, a_{n-1}]$ and $\bbf = [b_0, b_1, \cdots, b_{n-1}]$, their aperiodic cross-correlation is defined as
\ben \label{eq:acorr}
R_{ab}(\tau) = \sum_{i}a_i \overline{b}_{i-\tau}, \quad 1-n\leq \tau \leq n-1,
\een
where $a_i=b_i=0$ if $i < 0$ or $i \geq n$, and the overbar represents the complex conjugation. When $\abf = \bbf$, $R_{ab}(\tau)$ is abbreviated as the aperiodic auotcorrelation $R_a(\tau)$.

Similarly, the periodic autocorrelation of $\abf$ is defined as: 
\ben \label{eq:pcorr}
C_{a}(\tau) = \sum_{i=0}^{n-1}a_{(i)} \overline{a}_{(i-\tau)}, \quad \tau \in \Znum_n
\een
where the subscript $(\cdot)$ represents the index modulo $n$. $C_{a}(\tau)$ is related to $R_{a}(\tau)$ by the following equation:
\ben \label{eq:ac_and_pc}
C_{a}(\tau) = R_{a}(\tau) + R_{a}(\tau-n).
\een

Denote by $\abf^*$ the flipped and conjugate version of $\abf$, i.e., $\abf^* = \left[\overline{a}_{n-1}, \cdots, \overline{a}_1, \overline{a}_0 \right]$. Then by the commutativity of multiplication of complex numbers, we have
\ben \label{eq:complex_normal}
R_{a^*}(\tau) = R_{a}(\tau).
\een

For a complex sequence $\abf$ of length $n$, define the polynomials 
\ben \label{eq:seq_pol}
a(z) \triangleq \sum_{i=0}^{n-1} a_i z^i,\quad a^*(z) \triangleq  \sum_{i=0}^{n-1} \overline{a}_{n-1-i} z^i,
\een
\ben \label{eq:corr_pol}
R_a(z)  \triangleq   \sum_{\tau=1-n}^{n-1} R_a(\tau) z^{\tau}.
\een
It is straightforward to verify that
\ben \label{eq:norm_and_corr}
a(z)a^*(z) = R_a(z)z^{n-1}.
\een 

\begin{Definition} \label{def:GCS}
Define the weight of a sequence $\abf$ of length $n$ as $w(\abf) \triangleq \sum_{i=0}^{n-1}a_i\overline{a}_i$. A set of sequences $\{\abf_1, \abf_2, \cdots, \abf_L\}$ with unimodular or zero entries and of lengths $n_1, n_2, \cdots,$ $n_L$ is called a Golay complementary sequence (GCS) set of cardinality $L$ if
\ben \label{ACondition}
\sum_{l=1}^{L}R_{a_l}(\tau) = \begin{cases}
    \sum_{l=1}^{L} w(\abf_l), &\tau=0\\
    0, & \tau \neq 0
\end{cases}.
\een
Let $n \triangleq \max{\{n_1, n_2, \cdots, n_L\}}$, and append zeroes to the sequences $\abf_1, \abf_2, \cdots, \abf_L$ to obtain $\abf^{\prime}_1, \abf^{\prime}_2, \cdots, \abf^{\prime}_L$ of equal length $n$, respectively. Then \eqref{ACondition} is equivalent to 
\ben \label{PCondition}
\sum_{l=1}^{L}a^{\prime}_l(z){a^{\prime}_l}^{*}(z) = \sum_{l=1}^{L} w(\abf_l) z^{n-1}.
\een
\end{Definition}
 Particularly, the GCS sets of cardinalities $L=2, 4, 8$ are referred to as GCS {\em pair}, GCS {\em quad}, and GCS {\em octet}, respectively. A sequence with entries restricted to the $M$-th unit roots is referred to as an $M$-phase sequence. For example, a 2-phase sequence has entries $\{1, -1\}$ and a 4-phase sequence has entries $\{1, -1, i, -i\}$ where $i=\sqrt{-1}$. A polyphase GCS pair must have equal sequence length, while for polyphase GCS sets of cardinality greater than $2$, the sequence lengths may be different.
 
 The following examples of GCS pairs are taken from \cite{golay1961complementary, craigen2002complex}, where the superscripts represent the sequence lengths.:
\ben \label{gcs:2}
\abf_1^{(2)} = [1, 1], \quad \abf_2^{(2)} = [1, -1];
\een
\bea \label{gcs:10}
\abf_1^{(10)} &= [1, -1, -1, 1, -1, 1, -1, -1, -1, 1], \\ \abf_2^{(10)} &= [1, -1, -1, -1, -1, -1, -1, 1, 1, -1];
\eea
\bea \label{gcs:26}
\abf_1^{(26)} = &[-1,1,-1,-1,1,1,-1,1,1,1,1,-1,-1,-1,\\
&-1,-1,-1,-1,1,1,-1,-1,-1,1,-1,1], \\
\abf_2^{(26)} = &[-1,1,-1,-1,1,1,-1,1,1,1,1,-1,1,\\
&-1,1,1,1,1,-1,-1,1,1,1,-1,1,-1];
\eea
\ben \label{gcs:3}
\abf_1^{(3)} = [1, 1, -1], \quad \abf_2^{(3)} = [1, i, 1];
\een
\ben \label{gcs:5}
\abf_1^{(5)} = [i, i, 1, -1, 1], \quad \abf_2^{(5)} = [i, 1, 1, i, -1];
\een
\bea \label{gcs:11}
\abf_1^{(11)} &= [1, i, -1, 1, -1, i, -i, -1, i, i, 1], \\ \abf_2^{(11)} &= [1, 1, -i, -i, -i, 1, 1, i, -1, 1, -1];
\eea
\bea \label{gcs:13}
\abf_1^{(13)} &= [1, 1, 1, i, -1, 1, 1, -i, 1, -1, 1, -i, i], \\ \abf_2^{(13)} &= [1, i, -1, -1, -1, i, -1, 1, 1, -i, -1, 1, -i];
\eea

Note that $\{\abf_1^{(1)}=1, \abf_2^{(1)} = 1\}$ is also a GCS pair. It is referred to as the trivial GCS pair while the others are referred to as the nontrivial GCS pairs.

The following recursive construction of polyphase GCS pair was due to Craigen \cite{craigen2002complex}.
\begin{Proposition} [\cite{craigen2002complex}] \label{prop:quater_seq}
    Given a nontrivial 2-phase GCS pair $\{\abf, \bbf\}$ with sequence length $s$, two polyphase GCS pairs $\{\cbf, \dbf\}$ and $\{\ebf, \fbf\}$ with sequence lengths $t$ and $u$ respectively, and the recursive compositions
    \bea \label{eq.4phaseGolay}
    \pbf &= \frac{1}{4} {\left[\abf+\bbf+\left(\bbf^*-\abf^*\right)\right]},\
    \qbf = \frac{1}{4} {\left[\abf+\bbf-\left(\bbf^*-\abf^*\right)\right]},\\
    \xbf &= \pbf \otimes \cbf + \qbf \otimes \dbf,\quad
    \ybf = \qbf^* \otimes \cbf - \pbf^* \otimes \dbf,\\
    \gbf &= \xbf \otimes \ebf + \ybf \otimes \fbf,\quad
    \hbf = \ybf^* \otimes \ebf - \xbf^* \otimes \fbf,
    \eea
    the so-obtained $\{\gbf, \hbf\}$ is a polyphase GCS pair of length $stu$.
  \end{Proposition}
  The lengths of 2-phase (4-phase) GCS pairs are referred to as 2-phase ($4$-phase) Golay numbers. Proposition \ref{prop:quater_seq} essentially states that if $s$ is a non-trivial 2-phase Golay number and $t, u$ are two 4-phase Golay numbers, then $stu$ is also a 4-phase Golay number, which leads to the following existence of 4-phase GCS pairs:
  
\begin{Corollary} [\cite{craigen2002complex}] \label{Qsize}
There exist 4-phase GCS pairs if the sequence length
\ben 
\begin{split}
 n \in {\cal G}_{4p} \triangleq &   \left\{ 2^{a+u}3^{b}5^{c}11^{d}13^{e} | a, b, c, d, e, u \in {\mathbb Z}^+ \cup \{0\},  \right.\\
 &\left.\quad   b+c+d+e \leq a+2u+1, u \leq c+e \right\}.
\end{split}
\een 
\end{Corollary}

We refer to ${\cal G}_{4p}$ as the set of 4-phase Golay numbers in this paper. 

The following corollary is easy to validate based on Corollary 1.

\begin{Corollary} \label{coro:product_golay}
    The product of $k$ $4$-phase Golay numbers is of form $2^{a+u}3^{b}5^{c}11^{d}13^{e}$, where $a, b, c, d, e, u \in {\mathbb Z}^+ \cup \{0\}, b+c+d+e \leq a+2u+k, u \leq c+e$.
\end{Corollary}

\begin{remark}
    Corollary \ref{coro:product_golay} states that the constraints upon the exponents of the products of $4$-phase Golay numbers are more relaxed than those of a 4-phase Golay number. In other words, ${\cal G}_{4p} \subsetneq {\cal G}_{4p} \cdot {\cal G}_{4p} \subsetneq {\cal G}_{4p}\cdot {\cal G}_{4p} \cdot {\cal G}_{4p} \subsetneq \ldots$. 
    For example, although $9$ is not a $4$-phase Golay number, we have $9=3\times 3$ as a product of two $4$-phase Golay number. This observation would result in more flexible lengths in this work.
\end{remark}

\begin{Definition} \label{def:CBS}
A 4-phase GCS quad $\{\abf_1, \abf_2, \abf_3, \abf_4\}$, with $\abf_1, \abf_2$ of length $n_1=n_2=s_1$ and $\abf_3, \abf_4$ of length $n_3=n_4=s_2$ ($s_1$ and $s_2$ may be different), is referred to as a complex base sequence (CBS) and is denoted by $CBS(s_1, s_2)$.
\end{Definition}
Obviously, $\{\abf_1, \abf_2, \abf_3, \abf_4\}$ is a CBS if $\{\abf_1, \abf_2\}$ and $\{\abf_3, \abf_4\}$ are two 4-phase GCS pairs. 
There also exist $CBS(s_2+1, s_2)$ if 
\cite{ craigen2002complex, djokovic2010base} 
 \ben \label{eq:B}
 2s_2+1 \in \mathcal{B} \triangleq \{2b+1\, \vert \, 1\leq b \leq 38\} \cup \left(\{2\}\cdot\mathcal{S}_1+\{1\}\right),
 \een
  where 
 \bea
 \mathcal{S}_1 \triangleq \{0\}\cup {\cal G}_{4p}. 
 \eea
E.g., the following $\{\abf_1^{(8)}, \abf_2^{(8)}, \abf_3^{(7)}, \abf_4^{(7)}\}$ is a $CBS(8, 7)$:
\bea \label{eq:cbs87}
\abf_1^{(8)} &= [-1,1,1,1,1,1,-1,1], \\
\abf_2^{(8)} &= [1,1,1,-1,-1,1,-1,1],\\
\abf_3^{(7)} &= [-1,1,1,-1,1,1,1], \\
\abf_4^{(7)} &= [1,-1,1,1,1,-1,-1],
\eea 
and $\{\abf\vert 1, \abf\vert -1, \bbf, \bbf\}$ is a $CBS(g+1, g)$ if $\{\abf, \bbf\}$ is a $4$-phase GCS pair of length $g$.

\subsection{Perfect Sequences over Signed Symmetric Group} \label{sec:pre.spv}
The notion of signed symmetric group was introduced in \cite{craigen1995signed} to improve the asymptotic existence of Hadamard matrices.

\begin{Definition}
    A signed symmetric group denoted by $SP_v$ is a group of order $2^v v!$, isomorphic to the multiplicative group of all $v\times v$ signed permutation matrices, which have exactly one nonzero entry $\pm 1$ in each row and each column.
\end{Definition}

For example, $SP_2 = \langle i, j \vert i^2=-1, j^2=1, ij=-ji \rangle = \{\pm 1, \pm i, \pm j, \pm ij\}$ is a signed symmetric group of order $8$, identified with the group of all $2\times 2$ signed permutation matrices by the following isomorphism:
\ben \label{eq:isomorphism}
1 \mapsto \begin{bmatrix}
    1 & 0 \\ 0 & 1
\end{bmatrix}, \quad
i \mapsto \begin{bmatrix}
    0 & -1 \\ 1 & 0
\end{bmatrix}, \quad
j \mapsto \begin{bmatrix}
    1 & 0 \\ 0 & -1
\end{bmatrix}.
\een

Note that $S_{\Cnum} \triangleq \langle i\vert i^2=-1 \rangle = \{\pm 1, \pm i\}$ is an abelian subgroup of $SP_2$, and the field of complex numbers $\Cnum $ can be viewed as a group ring $\Rnum[S_{\Cnum}] \triangleq \{x_1 \cdot 1 + x_2 \cdot -1 + x_3 \cdot i + x_4 \cdot -i\, \vert\, x_1, x_2, x_3, x_4 \in \Rnum\}$ with the complex conjugation identified with the transpose of the $2\times 2$ matrix representation, i.e., 
\ben
\overline{i} = -i\ \mapsto\ \begin{bmatrix}
    0 & -1 \\ 1 & 0
\end{bmatrix}^T = \begin{bmatrix}
    0 & 1 \\ -1 & 0
\end{bmatrix}.
\een

Based on this observation, for sequences with entries in the signed symmetric group ring $\Rnum[SP_v]$, it is natural to define their correlations and polynomials as in \eqref{eq:acorr}, \eqref{eq:pcorr}, \eqref{eq:seq_pol} and \eqref{eq:corr_pol}, with the complex conjugation generalized to the matrix transpose, e.g., in $SP_2$, $\overline{i} = -i$ and $\overline{j} = j$. 

One may verify that for sequences over $\Rnum[SP_v]$ the properties \eqref{eq:ac_and_pc} and \eqref{eq:norm_and_corr} still hold. But \eqref{eq:complex_normal} may fail, since $R_{a}(\tau) = \sum_i a_{i}\overline{a}_{i-\tau} $ and $
R_{a^*}(\tau) = \sum_i \overline{a}_{i-\tau} a_{i}$ but
$a_{i}\overline{a}_{i-\tau} \ne  \overline{a}_{i-\tau}a_{i}$ in general because the multiplication of two signed permutation matrices is not necessarily commutative.

Define the support of a sequence $\abf$ of length $n$ as $supp(\abf) \triangleq \{i \, \vert\, a_i \neq 0, i \in \Znum_n\}$. A series of sequences are referred to be disjoint if the supports of any two sequences are disjoint, and to be supplementary if they are disjoint and the union of their supports is $\Znum_n$, e.g., $\abf = [1, 0, -1, 0]$, $\bbf = [0, i, 0, 0]$ and $\cbf = [0, 0, 0, j]$ are supplementary. A sequence $\abf$ is referred to be quasi-symmetric if its support is symmetric, i.e.,  $i \in supp(\abf) \Leftrightarrow n-1-i \in supp(\abf)$, e.g., $\abf = [1, 0, 0, i, 0, 0, j]$ is quasi-symmetric.

For $a, b\in \Rnum[SP_v]$, we can use them to construct an element in $\Rnum[SP_{2v}]$ denoted by
\ben
\begin{bmatrix}
    \pm a & 0\\
    0 & \pm b
\end{bmatrix}\ or\ 
\begin{bmatrix}
    0& \pm a \\
    \pm b & 0
\end{bmatrix},
\een
which is identified with the matrix of order $2v$ obtained by replacing $a, b$ with their matrix representations and replacing $0$ with a $v\times v$ zero matrix, e.g., 
\ben
\begin{bmatrix}
    i & 0 \\ 0 & -j
\end{bmatrix} \mapsto \begin{bmatrix}
    0 & -1& 0 &0 \\
    1 & 0& 0 &0 \\
    0 & 0& -1 &0 \\
    0 & 0 & 0 & 1
\end{bmatrix}.
\een
And we can embed $\Rnum[SP_v]$ into $\Rnum[SP_{2v}]$ by the following identification:
\ben \label{eq:embedding}
\begin{bmatrix}
    a & 0\\
    0 & a
\end{bmatrix} = a
\begin{bmatrix}
    1 & 0\\
    0 & 1
\end{bmatrix} = a,\ \text{with}\ a \in SP_v, \begin{bmatrix}
    a & 0\\
    0 & a
\end{bmatrix} \in SP_{2v},
\een
which is important and would be used extensively in Section \ref{sec:hadamard}.

Finally, we give the definition of perfect sequence with entries in a signed symmetric group.
\begin{Definition}
    A sequence $\abf = [a_0, a_1, \cdots, a_{n-1}]$ with entries in $SP_v$ is called a perfect sequence over $SP_v$ if its periodic autocorrelation is a scaled delta function, i.e., 
    \ben
        C_{a}(\tau) = \begin{cases}
        n, &\tau=0\\
        0, & \tau \neq 0
    \end{cases}.
    \een
\end{Definition}
This definition turns out to be important for the construction of block-circulant Hadamard matrices in Section \ref{sec:hadamard}.

\section{Construction of GCS of Arbitrary Length} \label{SEC:gcs} 
In this section, we present a new construction of $4$-phase GCS sets of cardinality $2^{3+\lceil \log_2 r \rceil}$ with arbitrary sequence length $n$, where $r$ is the number of nonzero digits of the $10^{13}$-base expansion of $n$.

The construction is based on the following theorem, which obtains new GCS sets via multiplying and adding up the lengths of some known GCS sets.

\begin{theorem} \label{thm:compromise}
    Given a polyphase GCS set $\{\abf_1, \abf_2, \cdots, \abf_{2L}\}$ where the length of each sequence is $s$, and a polyphase GCS set $\{\bbf_1, \bbf_2, \cdots, \bbf_{2M}\}$ where the length of $\bbf_{2m-1}$ is $t$ while the length of $\bbf_{2m}$ is $u$ for $m=1, 2, \cdots, M$. Suppose 
    \ben \label{concate_zeros}
    \bbf_{2m-1}^{\prime} = \bbf_{2m-1}\,\vert\,{\bf{0}}^{(u)}, \quad
    \bbf_{2m}^{\prime} = {\bf{0}}^{(t)}\,\vert\,\bbf_{2m},
    \een
    where $\vert$ represents concatenating two sequences and ${\bf{0}}^{(u)}$ $({\bf{0}}^{(t)})$ represents the zero sequence of length $u\ (t)$. For $1\leq l \leq L$, $1\leq m\leq M$, let
    \bea \label{compromise}
    \cbf_{l, m} &= \abf_{2l-1} \otimes \bbf_{2m-1}^{\prime} + \abf_{2l} \otimes \bbf_{2m}^{\prime},\\
    \dbf_{l, m} &= \abf_{2l}^* \otimes \bbf_{2m-1}^{\prime} - \abf_{2l-1}^* \otimes \bbf_{2m}^{\prime}.
    \eea
    then $\{\cbf_{l, m}, \dbf_{l, m} \ \vert \ 1\leq l\leq L, 1\leq m\leq M\}$ is a polyphase GCS set of cardinality $2LM$, and the length of each sequence is $s(t+u)$.
\end{theorem}

To prove Theorem \ref{thm:compromise} we need to recall the following lemma established in \cite{du2023polyphase}.

\begin{lemma} \cite[Lemma 4]{du2023polyphase} \label{lem:compromise}
  Given a commutative ring ${\cal R}$ with an involution *, and $a_1,\cdots a_{2L}, b_1, \cdots, b_{2M} \in {\cal R}$. For $1\leq l\leq L, 1\leq m\leq M$, suppose
  \bea  \label{eq:composite}
  c_{l, m} &\triangleq a_{2l-1} b_{2m-1} + a_{2l}b_{2m},\\
  d_{l, m} &\triangleq a_{2l}^* b_{2m-1} - a_{2l-1}^*b_{2m},
  \eea
  then
  \ben \label{eq:norm}
  \sum_{l=1}^{L}\sum_{m=1}^{M} c_{l, m}c_{l, m}^{*} + d_{l, m}d_{l, m}^{*} = \sum_{l=1}^{2L} a_l a_l^* \sum_{m=1}^{2M} b_m b_m^*.
  \een
\end{lemma}

\begin{proof}[Proof of Theorem \ref{thm:compromise}]
    From \eqref{compromise}, we have
\bea \label{eq:compo_poly}
c_{l, m}(z) &= a_{2l-1}(z^{t+u}) b_{2m-1}^{\prime}(z) + a_{2l}(z^{t+u}) b_{2m}^{\prime}(z),\\
d_{l, m}(z) &= a_{2l}^*(z^{t+u}) b_{2m-1}^{\prime}(z) - a_{2l-1}^*(z^{t+u})b_{2m}^{\prime}(z).
\eea
Note that the polynomials of complex sequences constitute a commutative ring, and the map from $a(z)$ to $a^*(z)$ is an involution. Hence \eqref{eq:composite} in Lemma \ref{lem:compromise} is identified with \eqref{eq:compo_poly} with $a_{l} = a_{l}(z^{t+u})$ for $1\leq l \leq 2L$, $b_{m} = b_{m}^{\prime}(z)$ for $1\leq m \leq 2M$ and $c_{l, m}= c_{l, m}(z)$, $d_{l, m}= d_{l, m}(z)$ for $1\leq l \leq L$, $1\leq m \leq 2M$. Then we have the following complementarity of autocorrelations:
\bea
  & \sum_{l=1}^{L}\sum_{m=1}^{M} c_{l, m}(z)c_{l, m}^{*}(z) + d_{l, m}(z)d_{l, m}^{*}(z)\\
  \overset{(1)}{=}&\ \sum_{l=1}^{2L} a_l(z^{t+u}) a_l^*(z^{t+u}) \sum_{m=1}^{2M} b_{m}^{\prime}(z) b_m^{\prime *}(z)\\
  \overset{(2)}{=}&\ \sum_{l=1}^{2L} w(\abf_l) z^{(s-1)(t+u)} \sum_{m=1}^{2M} w(\bbf_m) z^{t+u-1} \\
  =&\ 2LMs(t+u) z^{s(t+u)-1},
\eea
where $\overset{(1)}{=}$ follows from \eqref{eq:norm} and $\overset{(2)}{=}$ holds since $\{\abf_1, \cdots, \abf_{2L}\}$ and $\{\bbf_1, \cdots, \bbf_{2M}\}$ are two GCS sets. Hence it follows from Definition \ref{def:GCS} that $\{\cbf_{l, m}, \dbf_{l, m} \ \vert \ 1\leq l\leq L, 1\leq m\leq M\}$ is a GCS set. It is obvious that the GCS set has length $s(t+u)$ and is of cardinality $2LM$. Besides, owing to the disjoint structure of $\bbf_{2m-1}^{\prime}$ and $\bbf_{2m}^{\prime}$, all the entries of $\cbf_{l,m}, \dbf_{l, m}$ are polyphase. 
\end{proof}

\begin{remark} \label{rem:1}
    Theorem \ref{thm:compromise} still holds even if $\{\bbf_{2m-1} \ \vert \ m=1, 2, \cdots, M\}$ or $\{\bbf_{2m} \ \vert \ m=1, 2, \cdots, M\}$ consists of empty sequences, which means either $t$ or $u$ can be zero.
\end{remark}


\begin{remark} \label{rem:ingredient}
    Two kinds of ingredients can be fed into Theorem \ref{thm:compromise}: the $4$-phase GCS pairs and the CBS, as shown in Section \ref{sec:gcs:pair} and Section \ref{sec:gcs:cbs}, respectively. Note that two GCS pairs can be grouped into a special CBS, although a CBS needs not be a combination of two GCS pairs.
\end{remark}

\subsection{Feasible Lengths of GCS Sets Constructed from GCS pairs} \label{sec:gcs:pair}
\begin{Corollary} \label{coro:gcs_set}
    Define $\mathcal{S}_k \triangleq \mathcal{S}_1 \cdot \left(\mathcal{S}_{k-1}+\mathcal{S}_{k-1}\right)$, $k\geq 2$. Using Theorem \ref{thm:compromise} we can construct $4$-phase GCS sets of cardinality $2^k$ where each sequence has the same length $n \in \mathcal{S}_k$. 
\end{Corollary}

\begin{proof} First, for $k=2$: In Theorem \ref{thm:compromise}, let $L=1$ and $M=2$, and $\{\abf_1, \abf_2\}$, $\{\bbf_1, \bbf_3\}$ and $\{\bbf_2, \bbf_4\}$ are $4$-phase GCS pairs with sequence lengths $s$, $t$ and $u$ respectively. Then $\{\cbf_{1,1}, \dbf_{1,1}, \cbf_{1,2}, \dbf_{1,2}\}$ is a $4$-phase GCS quad, and the length of each sequence is $s(t+u)$ $\in \mathcal{S}_1 \cdot \left(\mathcal{S}_{1}+\mathcal{S}_{1}\right) = \mathcal{S}_2$.

    Suppose we have constructed two $4$-phase GCS sets $\mathcal{A}_1$ and $\mathcal{A}_2$ of cardinality $2^k$ with sequence lengths in $\mathcal{S}_k$. Next apply Theorem \ref{thm:compromise} again: let $L=1$ and $M=2^{k}$, and $\{\bbf_1, \bbf_3, \cdots, \bbf_{2^{k+1}-1}\} = \mathcal{A}_1$, $\{\bbf_2, \bbf_4, \cdots, \bbf_{2^{k+1}}\} = \mathcal{A}_2$, and $\{\abf_1, \abf_2\}$ is a $4$-phase GCS pair. Then $\{\cbf_{1, m}, \dbf_{1, m} \,\vert\,1\leq m\leq 2^k\}$ is a GCS set of cardinality $2^{k+1}$, and each sequence has the same length $n \in \mathcal{S}_1 \cdot \left(\mathcal{S}_{k}+\mathcal{S}_{k}\right) = \mathcal{S}_{k+1}$. By induction, we have completed the proof.
\end{proof}
As an application example of Theorem \ref{thm:compromise}, we can construct a GCS quad with length $87$, since $87 = 3\times(3+26)$. Indeed, from \eqref{gcs:26}\eqref{gcs:3} we can construct ($\otimes$ has higher precedence than $\vert$)
\bea
\cbf_{1, 1} &= \abf_{1}^{(3)}\otimes \abf_{1}^{(3)}\ \big\vert\ \abf_{2}^{(3)}\otimes \abf_{1}^{(26)}\\
\cbf_{1, 2} &= \abf_{1}^{(3)}\otimes \abf_{2}^{(3)}\ \big\vert\ \abf_{2}^{(3)}\otimes \abf_{2}^{(26)}\\
\dbf_{1, 1} &= \left(\abf_{2}^{(3)}\right)^*\otimes \abf_{1}^{(3)} \bigg\vert \left(\underline{\abf_{1}^{(3)}}\right)^*\otimes \abf_{1}^{(26)}\\
\dbf_{1, 2} &= \left(\abf_{2}^{(3)}\right)^*\otimes \abf_{2}^{(3)} \bigg\vert \left(\underline{\abf_{1}^{(3)}}\right)^*\otimes \abf_{2}^{(26)}.
\eea

When $k=2$, Theorem \ref{thm:compromise} degenerates into the construction of GCS quad in \cite[Theorem 5]{du2023polyphase}. Hence the example given above is not new. But for $k=3$, the GCS octets constructed by Theorem \ref{thm:compromise} have more possible lengths, e.g., $127$, $199$, $281$, $283$, $\cdots$, which cannot be covered by the GCS quads. A more comprehensive comparison of their possible lengths are given below.

Let $\rho_k(n) \triangleq \abs{\{s_k \ \vert \ s_k \in \mathcal{S}_k, 1\leq s_k\leq n \}}$, and define the density of the GCS set as $D_k(n) \triangleq \frac{\rho_k(n)}{n}$, which is a measure of the richness of the lengths of the constructed GCS sets. By computational verifications, the densities of GCS pairs, GCS quads and GCS octets obtained from Corollary \ref{coro:gcs_set} for $1\leq n \leq 10^{10}$ are plotted by the solid lines with different markers in Fig. \ref{fig:density}. The densities of GCS sets of larger cardinalities decrease slower. In particular, $D_3(10^{10})$ is very closed to $1$, which means that almost all the lengths no greater than $10^{10}$ can be covered by the GCS octets obtained by Corollary \ref{coro:gcs_set} (the first length that cannot be covered is \numprint{5433479347}).

\begin{figure}[ht!]
\centering
\includegraphics[width=3.3in]{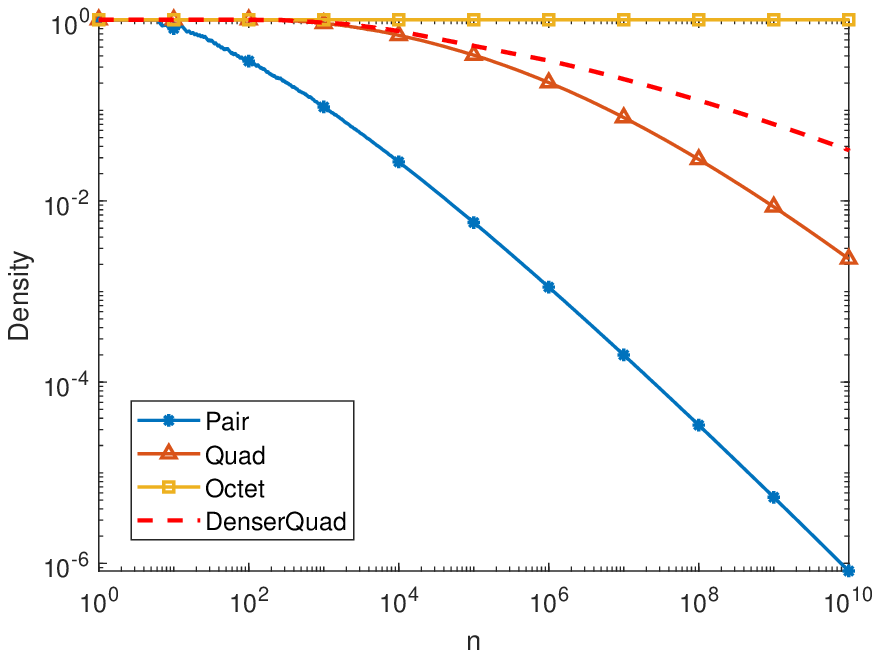}
\caption{Densities of $4$-phase GCS sets.}
\label{fig:density}
\end{figure}

From Fig \ref{fig:density}, it is tempting to conjecture that the GCS octets in Corollary \ref{coro:gcs_set} can cover almost all the positive integers. However, the following property shows that this is impossible.

\begin{Property} \label{property:density}
    For the GCS sets of finite cardinality $2^k$ in Corollary \ref{coro:gcs_set}, $\lim_{n \to 0} D_k(n) = 0$.
\end{Property}
\begin{proof}
    Let $\Tilde{\mathcal{S}}_1 \triangleq \{0\} \cup \{2^a3^b5^c11^d13^e \,\vert\, a, b, c, d, e\geq 0\}$ and define $\Tilde{\mathcal{S}}_k \triangleq \Tilde{\mathcal{S}}_{k-1} + \Tilde{\mathcal{S}}_{k-1}$, $k\geq 2$. Obviously $\mathcal{S}_k \subseteq \Tilde{\mathcal{S}}_k$ for $k\geq 1$. Let $\tilde{\rho}_k(n) \triangleq \abs{\{\Tilde{s}_k \ \vert \ \Tilde{s}_k \in \Tilde{\mathcal{S}}_k, 0\leq \Tilde{s}_k \leq n \}}$. Then we only need to prove that $\lim_{n\to \infty} \frac{\tilde{\rho}_k(n)}{n} = 0$. Note that $\tilde{\rho}_k(n) \leq \tilde{\rho}_{k-1}^2(n)$ and $\tilde{\rho}_1(n) \leq 1+\prod_{i\in\{2, 3, 5, 11, 13\}} \left(\lfloor \log_in\rfloor +1\right)$, hence $\tilde{\rho}_k(n) = \mathcal{O}\left((\log_{2}n)^{5\cdot 2^{k-1}}\right)$. Then we have $\lim_{n\to \infty}\frac{\tilde{\rho}_k(n)}{n} = 0$ for any finite $k$.
\end{proof}

\subsection{Denser Existence from Complex Base Sequences} \label{sec:gcs:cbs}
Note that in the proof of Corollary \ref{coro:gcs_set}, two GCS pairs $\{\bbf_1, \bbf_3\}$ and $\{\bbf_2, \bbf_4\}$ are combined as a CBS to feed into Theorem \ref{thm:compromise} for $k=2$. But as noted in Remark \ref{rem:ingredient}, a CBS needs not be a combination of two GCS pairs. Hence the $CBS(s_2+1, s_2)$ in Section \ref{sec:pre.gcs} with $2s_2+1 \in \mathcal{B}$ [c.f. \eqref{eq:B}], can also be fed into Theorem \ref{thm:compromise} for $k=2$ to obtain new lengths.

Besides, \cite[Theorem 3]{du2023polyphase} constructed $4$-phase GCS quads of equal length (a special case of CBS) by using Yang's multiplicative construction \cite{yang1989composition}, as reformulated in the following proposition with a slight modification that $\{\abf, \bbf, \cbf, \dbf\} \in CBS(s_1, s_2)$ but $\{\abf, \bbf\}$ and $\{\cbf, \dbf\}$ are not necessarily two GCS pairs, so that the construction can still proceed recursively.
\begin{Proposition} \label{prop:pre_quad}
    Given $\{\abf, \bbf, \cbf, \dbf\} \in CBS(s_1, s_2)$, and two GCS pairs $\{\ibf, \jbf\}$ and $\{\kbf, \lbf\}$ with sequence lengths $t_1$ and $t_2$, respectively. Construct a GCS quad $\{\ebf, \fbf, \gbf, \hbf\}$ as follows:
    \begin{enumerate}
    \item \bea \label{eq:concate}
  &\ebf = \abf\otimes \ibf \vert {\bf{0}}^{(t_2s_{2})} \vert {\bf{0}}^{(t_2s_{2})} \vert \bbf\otimes \jbf \\
  &\gbf = \bbf^*\otimes \ibf \vert {\bf{0}}^{(t_2s_{2})} \vert {\bf{0}}^{(t_2s_{2})} \vert {\underline{\abf}^*\otimes \jbf}\\
  &\fbf = {\bf{0}}^{(t_1s_{1})} \vert \cbf\otimes \kbf \vert \dbf\otimes \lbf \vert {\bf{0}}^{(t_1s_{1})} \\
  &\hbf = {\bf{0}}^{(t_1s_{1})} \vert \dbf^*\otimes \kbf \vert {\underline{\cbf}^*\otimes \lbf} \vert {\bf{0}}^{(t_1s_{1})},
  \eea
  where $t_1$ must equal $t_2$ if $\{\abf, \bbf\}$ and $\{\cbf, \dbf\}$ are not two GCS pairs;
        \item if $s_1 = s_2 +1$, we can also have
        \bea \label{eq:inter}
          &\ebf = \abf/{\bf{0}}^{(s_2)}, \qquad \gbf = \bbf/{\bf{0}}^{(s_2)}, \\ &\fbf = {\bf{0}}^{(s_1)}/\cbf, \quad \hbf = {\bf{0}}^{(s_1)}/\dbf,
        \eea
        where $/$ represents interleaving two sequences.
        
    \end{enumerate}
    Suppose $\{\ebf_1, \fbf_1, \gbf_1, \hbf_1\}$ and $\{\ebf_2, \fbf_2, \gbf_2, \hbf_2\}$ are two GCS quads from \eqref{eq:concate} or \eqref{eq:inter}, then the following $\{\pbf, \qbf, \rbf, \sbf\}$ is a $4$-phase GCS quad with equal sequence length:
    \bea 
  \pbf &= \ebf_1 \otimes \fbf_2^* - \gbf_1^* \otimes \ebf_2 + \fbf_1 \otimes \gbf_2 + \hbf_1 \otimes \hbf_2\\
  \qbf &= \ebf_1^* \otimes \ebf_2 + \gbf_1 \otimes \fbf_2^* - \fbf_1 \otimes \hbf_2^* + \hbf_1 \otimes \gbf_2^*\\
  \rbf &= \fbf_1^* \otimes \ebf_2 - \hbf_1 \otimes \fbf_2 + \ebf_1 \otimes \hbf_2^* + \gbf_1 \otimes \gbf_2\\
  \sbf &= -\fbf_1 \otimes \fbf_2 - \hbf_1^* \otimes \ebf_2 + \ebf_1 \otimes \gbf_2^* - \gbf_1 \otimes \hbf_2,
  \eea
\end{Proposition}

\begin{proof}
    The proof of Proposition \ref{prop:pre_quad} is the same as that in \cite{du2023polyphase}, except for the following reason why $t_1$ must equal $t_2$ if $\{\abf, \bbf\}$ and $\{\cbf, \dbf\}$ in \eqref{eq:concate} are not two GCS pairs. Similar to \cite[Proposition IV.2]{du2023polyphase}, one may prove that 
\bea
R_e(z) + R_g(z) &= \left(R_a(z^{t_1}) + R_b(z^{t_1})\right) \left(R_i(z) + R_j(z)\right) \\
&= 2t_1 \left(R_a(z^{t_1}) + R_b(z^{t_1})\right)
\eea
and 
\bea
R_f(z) + R_h(z) &= \left(R_c(z^{t_2}) + R_d(z^{t_2})\right) \left(R_k(z) + R_l(z)\right)\\
&= 2t_2 \left(R_c(z^{t_2}) + R_d(z^{t_2})\right).
\eea
If $\{\abf, \bbf\}$ and $\{\cbf, \dbf\}$ are two GCS pairs, then we have 
\ben
R_e(z) + R_g(z) = 4s_1t_1, \ R_f(z) + R_h(z) = 4s_2t_2,
\een
which is sufficient for the autocorrelations of $\{\ebf, \fbf, \gbf, \hbf\}$ being complementary; otherwise, the condition $t_1 = t_2$ is necessary so that 
\bea
&R_e(z) + R_g(z) + R_f(z) + R_h(z)\\
=& 2t_1 \left(R_a(z^{t_1}) + R_b(z^{t_1}) + R_c(z^{t_2}) + R_d(z^{t_2})\right)\\
=& 4t_1(s_1 + s_2).
\eea
\end{proof}

Feeding into Proposition \ref{prop:pre_quad} the $CBS(s_1,s_2)$ with $s_1 = s_2 +1$ and the $4$-phase GCS pairs in Section \ref{sec:pre.gcs}, we can construct $CBS(s, s)$ with $s\in \mathcal{F} \triangleq \mathcal{E}^2$, where $\mathcal{E} \triangleq \mathcal{B} \cup \left(\{2\} \cdot \mathcal{S}_1 \cdot \mathcal{B}\right) \cup \left(\{2\} \cdot (\mathcal{S}_1^2 + \mathcal{S}_1^2)\right)$ is the set of feasible lengths of $\{\ebf, \fbf, \gbf, \hbf\}$ in Proposition \ref{prop:pre_quad}: 
\begin{itemize}
    \item Because $\{\abf, \bbf, \cbf, \dbf\}$ in \eqref{eq:inter} $\in CBS(s_2+1, s_2)$, $\mathcal{B} \subset \mathcal{E}$.
    \item If $\{\abf, \bbf, \cbf, \dbf\}$ in \eqref{eq:concate} $\in CBS(s_2+1, s_2)$ and $\{\ibf, \jbf\}$ and $\{\kbf, \lbf\}$ are two GCS pairs with sequence lengths $t_1 = t_2 \in \mathcal{S}_1$ then $\{2\} \cdot \mathcal{S}_1 \cdot \mathcal{B} \subset \mathcal{E}$.
    \item If $\{\abf, \bbf\}$, $\{\cbf, \dbf\}$, $\{\ibf, \jbf\}$ and $\{\kbf, \lbf\}$ in \eqref{eq:concate} are four GCS pairs, then $\{2\} \cdot (\mathcal{S}_1^2 + \mathcal{S}_1^2) \subset \mathcal{E}$.
\end{itemize} 

The constructed $CBS(s, s)$ can be fed recursively into Proposition \ref{prop:pre_quad} to update $\mathcal{F}_{new} \triangleq \mathcal{E}_{new}^2$ with $\mathcal{E}_{new} \triangleq \mathcal{E} \cup \left(\{4\}\cdot \mathcal{S}_1 \cdot \mathcal{F}\right)$: If $\{\abf, \bbf, \cbf, \dbf\}$ in \eqref{eq:concate} $\in CBS(s, s)$ with $s\in \mathcal{F}$ and $\{\ibf, \jbf\}$ and $\{\kbf, \lbf\}$ are two GCS pairs with sequence lengths $t_1 = t_2 \in \mathcal{S}_1$, then $\{4\}\cdot \mathcal{S}_1 \cdot \mathcal{F} \subset \mathcal{E}_{new}$. Note that $\pbf, \qbf, \rbf$ and $\sbf$ have equal length, hence $\mathcal{S}_2$ is directly enlarged  without resorting to Theorem \ref{thm:compromise}. 

After enlarging $\mathcal{S}_2$, we proceed to the steps $k=3, 4,\cdots$ in the proof of Corollary \ref{coro:gcs_set}. Denote by $\mathcal{S}_k^D$ the enlarged set for $k\geq 2$ (the superscript $D$ stands for "denser"). By computational verifications, the density of $\mathcal{S}_2^D$ is plotted as the dashed line in Fig. \ref{fig:density}. Computer searches also verify that $\mathcal{S}_3^D$ covers all the lengths no greater than $10^{13}$, which is a great improvement over $\mathcal{S}_3$ covering all the lengths less than \numprint{5433479347}. The C\texttt{++} codes for verifying the densities of 
$\mathcal{S}_2^D$ and $\mathcal{S}_3^D$ are available online: \url{https://github.com/csrlab-fudan/gcs_hadamard}.

\subsection{Construction of 4-phase GCS of Arbitrary Length}
Next, we show how to construct $4$-phase GCS of arbitrary equal length, with much smaller cardinality than the known results in the literatures.


\begin{theorem} \label{thm:gcs_log}
    Suppose $\mathcal{S}_k^D$ covers all the lengths no greater than $N$ for some $k$. Define a set of integers 
    \ben \label{eq:s_below}
    \undertilde{\mathcal{S}} \triangleq \{2, 10, 26\} \cdot \mathcal{S}_1.
    \een
    Let $P \triangleq \max\{p\,\vert\,p\in \undertilde{\mathcal{S}}, p\leq N+1\}$. The $P$-base expansion of any integer $n$ is 
    \ben
    n = n_0+n_1P+\cdots+n_qP^q,
    \een
    where $q = \lfloor \log_{P}n\rfloor$ and $0\leq n_0, n_1, \cdots, n_q < P$. Define $\mathcal{I} \triangleq \{i \,\vert\,n_i\neq 0, i=0, 1, \cdots, q\}$ and $r \triangleq \abs{\mathcal{I}}$. Then there exists a GCS set of cardinality $2^{k+\lceil \log_2 r \rceil}$, where all the sequences are of length $n$.
\end{theorem}
 
To prove Theorem \ref{thm:gcs_log}, we need to first establish the following proposition, which is a  natural extension of Proposition \ref{prop:quater_seq}.

\begin{Proposition} \label{prop:multiplicative}
    Given two polyphase GCS sets $\{\abf_1, \cdots, \abf_{2L}\}$ and $\{\bbf_1, \cdots, \bbf_{2M}\}$ with sequence lengths $s$ and $t$ respectively, and a nontrivial 2-phase GCS pair $\{\ebf, \fbf\}$ with sequence length $u$. Suppose
    \bea
    \pbf &= \frac{1}{4} {\left[\ebf+\fbf+\left(\fbf^*-\ebf^*\right)\right]},\\
    \qbf &= \frac{1}{4} {\left[\ebf+\fbf-\left(\fbf^*-\ebf^*\right)\right]},\\
    \bbf_{2m-1}^{\prime} &= \pbf \otimes \bbf_{2m-1} + \qbf \otimes \bbf_{2m},\\
    \bbf_{2m}^{\prime} &= \qbf^* \otimes \bbf_{2m-1} - \pbf^* \otimes \bbf_{2m},\\
    \cbf_{l, m} &= \abf_{2l-1} \otimes \bbf_{2m-1}^{\prime} + \abf_{2l} \otimes \bbf_{2m}^{\prime},\\
    \dbf_{l, m} &= \abf_{2l}^* \otimes \bbf_{2m-1}^{\prime} - \abf_{2l-1}^* \otimes \bbf_{2m}^{\prime}.
    \eea
    Then $\{\cbf_{l, m}, \dbf_{l, m} \, \vert \, 1\leq l\leq L, 1\leq m\leq M\}$ is a polyphase GCS set of cardinality $2LM$, and the length of each sequence is $stu$.
\end{Proposition}
In Proposition \ref{prop:multiplicative}, if let $L=1$, $\{\abf_1, \abf_2\}$ be a $4$-phase GCS pair with sequence length $s \in \mathcal{S}_1$ and $u\in \{2, 10, 26\}$, then with the same cardinality $2M$, the sequence length of $\{\cbf_{1, m}, \dbf_{1, m} \, \vert \, 1\leq m\leq M\}$ is enlarged by a factor of $su \in \undertilde{\mathcal{S}}$  [c.f. \eqref{eq:s_below}] compared with that of $\{\bbf_1, \cdots, \bbf_{2M}\}$. Hence the corollary follows in the below. 
\begin{Corollary} \label{coro:multiply}
     For a GCS set with sequences of equal length, the length can be enlarged by a factor of $\undertilde{s} \in \undertilde{\mathcal{S}}$ without enlarging the cardinality.
\end{Corollary}

Now we are ready to prove Theorem \ref{thm:gcs_log}. 
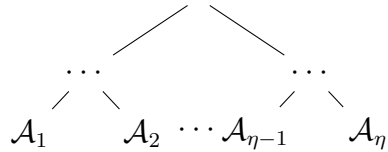
\begin{figure}
    \centering
    \begin{tikzpicture}[level distance=1.3cm,
       level 1/.style={sibling distance=3cm, level distance=1cm},
       level 2/.style={sibling distance=1.5cm, level distance=0.8cm}] 
    \node {}
       child {node {$\cdots$}
       child {node {$\mathcal{A}_1$}}
       child {node {$\mathcal{A}_2$}}
    }
    child {node {$\cdots$}
       child {node {$\mathcal{A}_{\eta-1}$}}
       child {node {$\mathcal{A}_{\eta}$}}
    };
    \node at (0,-1.75) {$\cdots$};
    \end{tikzpicture}
    \caption{Hierarchical applications of Theorem \ref{thm:compromise} in $\lceil \log_2 r \rceil$ levels.}
    \label{fig:hier}
\end{figure}
\begin{proof}[Proof of Theorem \ref{thm:gcs_log}]
    Because $0<n_i<P$ for $\forall\, i \in \mathcal{I}$ and $P\leq N+1$, we have $0<n_i\leq N$ for $\forall\, i \in \mathcal{I}$, which can be covered by $\mathcal{S}_k^D$ by assumption. By recursively applying Corollary \ref{coro:multiply}, the sequence length $n_i$ can be multiplied with $P^i$ without enlarging the cardinality.
 
    Now we have obtained $r$ GCS sets of cardinality $2^k$ denoted by $\mathcal{A}_1, \cdots, \mathcal{A}_r$, where the sequence length in each set is $n_iP^{i}$, $i \in \mathcal{I}$. If $r$ is not a power of $2$, then append $\eta - r$ sets of $2^k$ empty sequences denoted by $\mathcal{A}_{r+1}, \cdots, \mathcal{A}_{\eta}$ with $\eta \triangleq 2^{\lceil \log_2 r \rceil}$. Next we add up the lengths $n_iP^i$ by using Theorem \ref{thm:compromise} hierarchically in $\lceil \log_2 r \rceil$ levels, as illustrated by Fig. \ref{fig:hier}, where a parent node is the result of applying Theorem \ref{thm:compromise} to two child nodes to add up their lengths (let $L=1$ and $\abf_1 = \abf_2 = 1$). Then we obtain a GCS set of cardinality $2^{k+\lceil \log_2 r \rceil}$, where the length of each sequence is $n = n_0+n_1P+\cdots+n_qP^q$.
\end{proof}


 
As we have computationally verified that all the positive integers no greater than $N = 10^{13}$ can be covered by $\mathcal{S}_3^D$, we set $P \triangleq \max\{p\,\vert\,p\in \undertilde{\mathcal{S}}, p\leq N+1\} = 10^{13}$, and we can construct a GCS set of cardinality $2^{3+\lceil \log_2 r \rceil}$ with arbitrary sequence length $n$, where $r$ is the number of non-zero digits of the $10^{13}$-base expansion of $n$.

Some direct constructions of GCS sets can also cover arbitrary length \cite{chen2017novel, ghosh2022direct} \footnote{\cite{ghosh2022direct} constructed 2D Golay complementary array (GCA), a multi-dimensional generalization of GCS. By applying \cite[Theorem 2]{du2023polyphase} to the sequences we construct here, high-dimensional GCAs of arbitrary array size can also be constructed at the cost of enlarging the cardinality.}. As shown in Table \ref{tab:gcs_length}, the cardinality of the GCS set constructed in \cite{ghosh2022direct} may increase linearly with respective to $n$, e.g., if $n$ is a prime number. The cardinality in \cite{chen2017novel} increases logarithmically to the base $2$, while the cardinality in this work increases logarithmically to the base $10^{13}$.
\begin{table*}[t]
\caption{Parameters of GCS set of arbitrary length}  \label{tab:gcs_length}
    \centering
    \begin{threeparttable}
    \begin{tabular}{|c|c|c|c|}
        \hline
            Construction & Sequence length & Cardinality & Number of phases\\
        \hline
        \cite{ghosh2022direct} & $n = \prod_{i=1}^\lambda n_i^{m_i}$, $n_i$ is prime, $m_i\geq 0$ &$\prod_{i=1}^\lambda n_i$ & $\text{lcm}(n_1, n_2, \cdots, n_\lambda)$\\ 
        \hline
        \cite{chen2017novel} & \makecell{$n = 2^{m-1} + \sum_{\alpha = 1}^{k-1} d_\alpha 2^{\pi \left(m-k+\alpha\right) -1} + d_0 2^v,$ \tnote{1} \\ $k<m,\ 0\leq v \leq m-k,\ d_a \in \{0, 1\}$} & $2^{k+1}$ & any even number\\
        \hline
        Theorem \ref{thm:gcs_log} & $n = \sum_{i=0}^{q} n_i P^{i}$, $P=10^{13}$, $n_i \in \Znum_P$ &$2^{3+\lceil \log_2 r \rceil}, r = \abs{\{i \,\vert\,n_i>0\}}$ & $4$\\
        \hline
    \end{tabular}
    \begin{tablenotes}
        \item[1] $\pi$ is a permutation of $\{1, 2, \cdots, M\}$ satisfying some constraints \cite{chen2017novel}.
    \end{tablenotes}
\end{threeparttable}
\end{table*}

\section{Construction of Hadamard Matrices from GCS} \label{sec:hadamard}
An $n\times n$ matrix $\Hbf$ with entries in $\{1, -1\}$ is a Hadamard matrix of order $n$ if $\Hbf\Hbf^T = n\Ibf$. The celebrated Hadamard conjecture states that there may exist Hadamard matrices of order $4n$ for any positive integer $n$. Many efforts have been devoted to the Hadamard conjecture, see \cite{seberry2020hadamard} and the references therein for this topic. Here we are only interested in the constructions of Hadamard matrices from some GCS sets.

In \cite{craigen2002complex}, a Hadamard matrix of order $8n$ was constructed from a $4$-phase GCS quad with sequence length $n$, using a modification of the Goethal-Seidel array \cite{goethals1970skew}. The construction is reproduced in the below  for ease of reference.
\begin{Proposition}[\cite{craigen2002complex}] \label{prop:gs}
    First construct four circulant matrices $\Abf$, $\Bbf$, $\Cbf$ and $\Dbf$ from a $4$-phase GCS quad $\{\abf, \bbf, \cbf, \dbf\}$ with sequence length $n$. Then construct a matrix $\Tilde{\Hbf}$ over $SP_2$ of order $4n$ as follows:
    \ben
    \Tilde{\Hbf} = \begin{bmatrix}
        \Abf & -\Bbf\Rbf j & -\Cbf\Rbf j & -\Dbf\Rbf j & \\
        \Bbf\Rbf j & \Abf &-\Dbf^*\Rbf j & \Cbf^*\Rbf j \\
        \Cbf\Rbf j & \Dbf^* \Rbf j & \Abf & -\Bbf^*\Rbf j \\
        \Dbf\Rbf j & -\Cbf^*\Rbf j & \Bbf^*\Rbf j & \Abf
    \end{bmatrix},
    \een
    where 
    \ben
    \Rbf = \begin{bmatrix}
        0 & \cdots, & 0 & 1 \\
        0 & \cdots, & 1 & 0 \\
        \vdots & \reflectbox{$\ddots$} & \vdots & \vdots \\
        1 & \cdots, & 0 & 0 \\
    \end{bmatrix}.
    \een
    Finally replace the entries of $\Tilde{\Hbf}$ with the following $2\times 2$ matrices to obtain a Hadamard matrix $\Hbf$ of order $8n$:
   \bea
    &\pm 1 \mapsto \pm \begin{bmatrix}
        1 & 1 \\ 1 & -1
    \end{bmatrix}, \quad
    &\pm i \mapsto \pm \begin{bmatrix}
        -1 & 1 \\ 1 & 1
    \end{bmatrix}, \\
    &\pm j \mapsto \pm \begin{bmatrix}
        1 & 1 \\ -1 & 1
    \end{bmatrix}, \quad
    & \pm ij \mapsto \pm \begin{bmatrix}
    1 & -1 \\ 1 & 1
    \end{bmatrix}.
    \eea
    
\end{Proposition}
Based on this construction and the GCS quads constructed in Section \ref{SEC:gcs}, we have the following corollary.
\begin{Corollary} \label{coro:hadamard_quad}
    There exists Hadamard matrices of order $8n$ for any $n \in \mathcal{S}_2^D$.
\end{Corollary}

Similar to the sparsity of $\mathcal{S}_2$ in Property \ref{property:density}, one may prove the sparsity of $\mathcal{S}_2^D$ (numerically illustrated by the dashed line in Fig. \ref{fig:density}). Hence the existence of Hadamard matrices in Corollary \ref{coro:hadamard_quad} is far from the Hadamard conjecture. 

Note that the $4$-phase GCS sets constructed in Section \ref{SEC:gcs}
have arbitrary length, one may expect to utilize them to construct more Hadamard matrices. However, it is unknown whether a Hadamard matrix of order $2Ln$ can be constructed from a $4$-phase GCS set of cardinality $L>4$ with sequence length $n$. 

Alternatively, we follow the path of \cite{craigen1995signed} to construct Hadamard matrices by using the representation theory of signed symmetric group, and propose a construction different from the original method in \cite{craigen1995signed}. In doing so, we improve the asymptotic existence of Hadamard matrices, as summarized in the following theorem.

\begin{theorem} \label{thm:main}
    There exist block-circulant Hadamard matrices of order $2^t m$ with block size $2^{t-2}$ for any odd number $m$, where $t= 6\lfloor \frac{1}{40}\log_{2}m\rfloor + 10$.
\end{theorem}
In comparison, the best known result was $t=6\lfloor \frac{1}{26}\log_2\frac{m-1}{2}\rfloor +11$ \cite{livinskyi2012asymptotic}. The asymptotic results ignoring the floor function are compared in Fig. \ref{fig:asymptotic}. Note that the yellow line with $t=2$ corresponds with the Hadamard conjecture: if $t = 2$ for any odd number $m$, then there exist Hadamard matrices $\Hbf$ of order $4m$ for any odd number $m$, and by Sylwester's construction \cite{seberry2020hadamard} a Hadamard matrix of order $4\cdot (2m)$ can be constructed as 
\ben
\begin{bmatrix}
    \Hbf & \Hbf \\
    \Hbf & -\Hbf
\end{bmatrix},
\een
thus there exist Hadamard matrices of order $4n$ for any positive integer $n$.

\begin{figure}[ht!]
\centering
\includegraphics[width=3.3in]{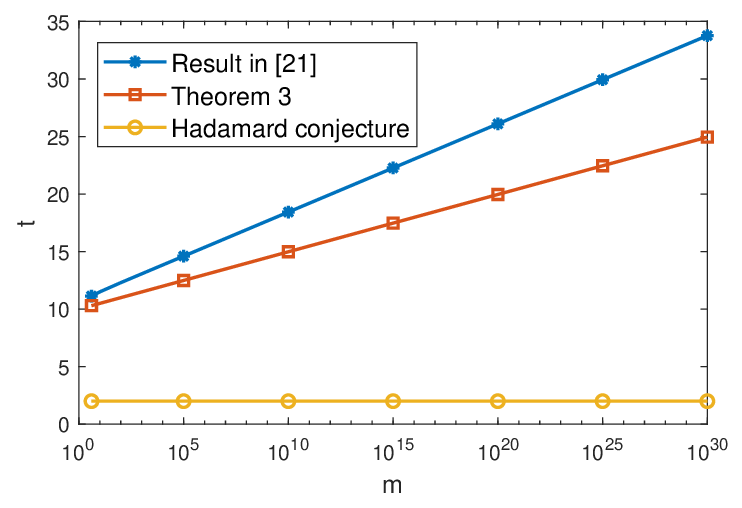}
\caption{Asymptotic Existence of Hadamard Matrices. (Note that $m$ is odd.)}
\label{fig:asymptotic}
\end{figure}

The remaining part of this section is devoted to the proof of Theorem \ref{thm:main}. In Section \ref{subsec:bridge}, we introduce a connection between a block-circulant Hadamard matrix and a perfect sequence over a signed symmetric group; In Section \ref{subsec:seq_combine}, we propose a method to combine two sequences into one sequence and preserve the autocorrelation property; In Section \ref{subsec:input} we construct a series of periodic complementary sequences with specific structure; In Section \ref{subsec:main} we collect the results in Section \ref{subsec:bridge}, Section \ref{subsec:seq_combine} and Section \ref{subsec:input} to complete the proof of Theorem \ref{thm:main}.

\subsection{Hadamard Matrix from Perfect Sequence} \label{subsec:bridge}
\cite[Theorem 5]{craigen1995signed} built a connection between a perfect sequence over a signed symmetric group and a block-circulant Hadamard matrix, as reformulated in the following lemma.
\begin{lemma}[\cite{craigen1995signed}] \label{lem:sg_hadamard}
    If there exists a perfect sequence $\cbf$ over $SP_v$ of length $n$ and $v$ is a known order of Hadamard matrices, then there exists a block-circulant Hadamard matrix of order $vn$ with block size $v$.
\end{lemma}
\begin{proof}[Sketch of proof]
    According to the method in \cite{craigen1995signed}, we can construct a circulant matrix $\Cbf$ over $SP_{v}$ of order $n$ from $\cbf$; then substitute the entries of $\Cbf$ with their representations of signed permutation matrices of order $v$ to obtain a $\{0, \pm 1\}$-matrix $\Dbf$ of order $vn$; finally let $\Hbf = \Dbf \left(\Ibf_{n} \otimes \Hbf_{v}\right)$ where $\Hbf_{v}$ is a known Hadamard matrix of order $v$, and we obtain $\Hbf$ as a Hadamard matrix of order $vn$.
\end{proof}

Motivated by this connection, we aim at constructing perfect sequences over $SP_{v}$ of flexible lengths $n$ --- of course $v$ should be as small as possible.

\subsection{Sequences Combining} \label{subsec:seq_combine}
First, we construct a sequence $\cbf$ over $\{0\}\cup SP_{2v}$ from two sequences $\abf$ and $\bbf$ over $\{0\}\cup SP_{v}$, so that $R_{c}(\tau) = R_{c^*}(\tau) = R_{a}(\tau) + R_{b}(\tau)$.

\begin{theorem} \label{thm:seq_combine}
    Suppose $\abf$ and $\bbf$ are two disjoint and quasi-symmetric sequences over $\{0\}\cup SP_{v}$ of length $n$, and 
    \ben \label{cond:normal} R_{a}(\tau) = R_{a^*}(\tau),\ R_{b}(\tau) = R_{b^*}(\tau); \een
    \ben \label{cond:commu} a_i b_j = b_j a_i,\ \forall\, i, j \in \Znum_n. 
    \een
    Construct two sequences $\xbf$ and $\ybf$ of length $n$ as follows:
    \ben \label{eq:sp_2m}
    x_i = \begin{bmatrix}
        a_i & 0 \\ 0 & \overline{a}_{n-i-1}
    \end{bmatrix}, \ 
    y_i = \begin{bmatrix}
        0 & b_i \\ -\overline{b}_{n-i-1} & 0
    \end{bmatrix},\ i \in \Znum_n.
    \een
    Let $\cbf = \xbf+\ybf$, then $\cbf$ is a quasi-symmetric sequence over $\{0\} \cup SP_{2v}$ with $supp(\cbf) = supp(\abf) \cup supp(\bbf)$, and
    \ben \label{eq:corr_combine}
    R_{c}(\tau) = R_{c^*}(\tau) = R_{a}(\tau) + R_{b}(\tau).
    \een
\end{theorem}

An example of Theorem \ref{thm:seq_combine}: let $\abf = [0, -i, 0, -1, 0]$, $\bbf = [1, 0, 0, 0, -i]$, then 
\ben
c_0 = \begin{bmatrix}
    0 & b_0 \\ -\overline{b}_4 & 0
\end{bmatrix} = \begin{bmatrix}
    0 & 1 \\ -i & 0
\end{bmatrix} = \begin{bmatrix}
    0 & 0 & 1 & 0 \\
    0 & 0 & 0 & 1 \\
    0 & 1 & 0 & 0 \\
    -1 & 0 & 0 & 0
\end{bmatrix},
\een
\ben
c_1 = \begin{bmatrix}
    a_1 & 0 \\ 0 & \overline{a}_3
\end{bmatrix} = \begin{bmatrix}
    -i & 0 \\ 0 & -1
\end{bmatrix} = \begin{bmatrix}
    0 & 1 & 0 & 0 \\
    -1 & 0 & 0 & 0 \\
    0 & 0 & -1 & 0 \\
    0 & 0 & 0 & -1
\end{bmatrix},
\een
\ben
c_2 = \begin{bmatrix}
    a_2 & 0 \\ 0 & \overline{a}_2
\end{bmatrix} = \begin{bmatrix}
    0 & 0 \\ 0 & 0
\end{bmatrix}  = \begin{bmatrix}
    0 & 0 & 0 & 0 \\
    0 & 0 & 0 & 0 \\
    0 & 0 & 0 & 0 \\
    0 & 0 & 0 & 0
\end{bmatrix},
\een
\ben
c_3 = \begin{bmatrix}
    a_3 & 0 \\ 0 & \overline{a}_1
\end{bmatrix} = \begin{bmatrix}
    -1 & 0 \\ 0 & i
\end{bmatrix} = \begin{bmatrix}
    -1 & 0 & 0 & 0 \\
    0 & -1 & 0 & 0 \\
    0 & 0 & 0 & -1 \\
    0 & 0 & 1 & 0
\end{bmatrix},
\een
\ben
c_4 = \begin{bmatrix}
    0 & b_4 \\ -\overline{b}_0 & 0
\end{bmatrix} = \begin{bmatrix}
    0 & -i \\ -1 & 0
\end{bmatrix} = \begin{bmatrix}
    0 & 0 & 0 & 1 \\
    0 & 0 & -1 & 0 \\
    -1 & 0 & 0 & 0 \\
    0 & -1 & 0 & 0 \\
\end{bmatrix}.
\een
Obviously, $\cbf$ is a quasi-symmetric sequence over $\{0\} \cup SP_{4}$ with $supp(\cbf) = supp(\abf) \cup supp(\bbf)$. Besides, direct calculations show that 
    \ben \label{eq:Rab}
    R_{a}(\tau) + R_{b}(\tau) = [i, 0, i, 0, 4, 0, -i, 0, -i]
    \een
and
\ben \label{eq:Rc}
R_c(\tau) = R_c^*(\tau) = \begin{cases}
    \Mbf_1, & \tau = 0 \\
    \Mbf_2, & \tau = -1, -3 \\
    \Mbf_3, & \tau = -2, -4 \\
    -R_c(\tau), & 1\leq \tau \leq 4
\end{cases}\ ,
\een
where
\bea 
\Mbf_1 &= 4\begin{bmatrix}
    1 & 0 & 0 & 0 \\
    0 & 1 & 0 & 0 \\
    0 & 0 & 1 & 0 \\
    0 & 0 & 0 & 1
\end{bmatrix},\ 
\Mbf_2 = \begin{bmatrix}
    0 & 0 & 0 & 0 \\
    0 & 0 & 0 & 0 \\
    0 & 0 & 0 & 0 \\
    0 & 0 & 0 & 0
\end{bmatrix}, \\ 
\Mbf_3 &= \begin{bmatrix}
    0 & -1 & 0 & 0 \\
    1 & 0 & 0 & 0 \\
    0 & 0 & 0 & -1 \\
    0 & 0 & 1 & 0
\end{bmatrix}.
\eea
According to the isomorphism in \eqref{eq:isomorphism} and the embedding in \eqref{eq:embedding}, \eqref{eq:Rab} is equal to \eqref{eq:Rc}.

The following lemma is crucial for proving Theorem \ref{thm:seq_combine}.
\begin{lemma} \label{lem:orthogonal}
    Given $a, b$ two elements in a ring satisfying $ab=ba$, $aa^* = a^*a$ and $bb^* = b^*b$ where $*$ is an anti-involution.
    Let
    \ben \label{eq:2_2_mat}
    \Cbf = \begin{bmatrix}
        a&b\\ -b^*&a^*
    \end{bmatrix}.
    \een
    Then $\Cbf$ is an orthogonal $2\times 2$ matrix, i.e., 
    \bea \label{eq:orthognal}
    \Cbf\Cbf^* =& \begin{bmatrix}
        a&b\\ -b^*&a^*
    \end{bmatrix} 
    \begin{bmatrix}
        a^*&-b\\ b^*&a
    \end{bmatrix} = (aa^*+bb^*) \begin{bmatrix}
        1&0\\0&1
    \end{bmatrix},\\
    \Cbf^*\Cbf =& 
    \begin{bmatrix}
        a^*&-b\\ b^*&a
    \end{bmatrix} \begin{bmatrix}
        a&b\\ -b^*&a^*
    \end{bmatrix} = (aa^*+bb^*) \begin{bmatrix}
        1&0\\0&1
    \end{bmatrix}.
    \eea
\end{lemma}
\begin{proof}[Proof of Theorem \ref{thm:seq_combine}]
    Because $\abf$ is quasi-symmetric, $a_i = \overline{a}_{n-i-1} = 0$ or $a_i \neq 0,  \overline{a}_{n-i-1} \neq 0$. Hence $x_i \in \{0\} \cup SP_{2v}$. Similarly $y_i \in \{0\} \cup SP_{2v}$. Note that $supp(\xbf) = supp(\abf)$ and $supp(\ybf) = supp(\bbf)$. Hence $\xbf$ and $\ybf$ are quasi-symmetric and disjoint as $\abf$ and $\bbf$. Then $\cbf = \xbf + \ybf$ is a quasi-symmetric sequence over $\{0\} \cup SP_{2v}$ and $supp(\cbf) = supp(\abf) \cup supp(\bbf)$. 

    Next we prove that $ R_{c}(\tau) = R_{c^*}(\tau) = R_{a}(\tau) + R_{b}(\tau)$ by specifying the ring in Lemma \ref{lem:orthogonal} as a polynomial ring with coefficients over a signed symmetric group ring, and the anti-involution as the map from $a(z)$ to $a^*(z)$. Note that 
    \ben
    c(z) = \begin{bmatrix}
        a(z) & b(z)\\ -b^*(z) & a^*(z)
    \end{bmatrix},
    \een
    and we have $a(z)b(z) = b(z)a(z)$ by \eqref{cond:commu}, $a(z)a^*(z) = a^*(z)a(z)$ and $b(z)b^*(z) = b^*(z)b(z)$ by \eqref{eq:norm_and_corr} and \eqref{cond:normal}. In Lemma \ref{lem:orthogonal}, let $a=a(z)$, $b=b(z)$, $\Cbf = c(z)$, then we have
    \ben \label{eq:norm_c}
    c(z)c^*(z) = c^*(z)c(z) = \left(a(z)a^*(z) + b(z)b^*(z)\right) \begin{bmatrix}
        1&0\\0&1
    \end{bmatrix}.
    \een
    By the embedding in \eqref{eq:embedding}, \eqref{eq:norm_c} is equivalent to 
    \ben
    c(z)c^*(z) = c^*(z)c(z) = a(z)a^*(z) + b(z)b^*(z),
    \een
    i.e., $ R_{c}(\tau) = R_{c^*}(\tau) = R_{a}(\tau) + R_{b}(\tau)$. 

\end{proof}

\begin{remark} \label{remark:corr_sum}
    The original construction in \cite[Lemma 2]{craigen1995signed} was presented in the language of matrices, while ours is presented in the language of sequences. Both constructions are based on the $2\times 2$ orthogonal matrix in Lemma \ref{lem:orthogonal}. In \cite[Lemma 2]{craigen1995signed}, the ring is specified as a matrix ring and the anti-involution transposes and conjugates a matrix, while our method specifies the ring as a polynomial ring and the anti-involution flips and conjugates a sequence. This leads to different definitions of "quasi-symmetric": in the context of \cite[Lemma 2]{craigen1995signed}, a sequence composing a circulant matrix is quasi-symmetric if the support of the remaining sequence is symmetric after deleting the first element, e.g., $\abf = [1, 0, 1, 1, 0]$ is quasi-symmetric; while in the context of this work, a sequence is quasi-symmetric if the support of the sequence is symmetric, e.g., $\abf = [0, 1, 1, 0]$ is quasi-symmetric. Besides, \cite[Lemma 2]{craigen1995signed} focuses on the periodic property while Theorem \ref{thm:seq_combine} focuses on the aperiodic property, which is more general according to \eqref{eq:ac_and_pc}.
\end{remark}

\subsection{Periodic Complementary Sequences with Quasi-symmetric and Supplementary Structure} \label{subsec:input}
Next, we construct a series of sequences over $\{0\} \cup S_{\Cnum}$ from $4$-phase GCS pairs and CBS, so that they can be fed sequentially into Theorem \ref{thm:seq_combine} to obtain a perfect sequence over a signed symmetric group.
\begin{theorem} \label{thm:disjoint_symmetric}
    Given $4k+8d$ sequences $\ebf_i, \fbf_i, \gbf_i, \hbf_i$, $i=1, 2, \cdots k+2d$, satisfying 
    \begin{enumerate}
        \item $\{\ebf_i, \fbf_i\}$ is a 4-phase GCS pair with sequence length $l_i$ for $\forall\, 1\leq i \leq k$, and $\{\ebf_{k+2i-1}, \fbf_{k+2i-1}, \ebf_{k+2i}, \fbf_{k+2i}\} \in CBS(s_{2i-1}, s_{2i})$ for $\forall\, 1\leq i \leq d$.
        \item $\{\gbf_i, \hbf_i\}$ is a 4-phase GCS pair for $\forall\, 1\leq i \leq k+2d$, where the length of $\gbf_i, \hbf_i$ is $m_i$ for $\forall\, 1 \leq i\leq k$, and $\gbf_{k+2i-1}, \hbf_{k+2i-1}, \gbf_{k+2i}, \hbf_{k+2i}$ must have the same length $t_i$ for $\forall\, 1\leq i \leq d$.
    \end{enumerate}
    Let 
    \ben \label{eq:total_len}
    n \triangleq 4\sum_{i=1}^{k} l_i m_i + 4\sum_{i=1}^{d} \left(s_{2i-1} + s_{2i}\right) t_i,
    \een
    \ben
    \lambda_i \triangleq
    \begin{cases}
        \sum\limits_{j=1}^{i-1} l_j m_j,& \ 1\leq i\leq k+1\\
        \lambda_{k+1}+\sum\limits_{j=1}^{i-k-1} s_{j}t_{\lceil j/2\rceil},&\ k+2\leq i \leq k+2d+1
    \end{cases}.
    \een
    Construct 
    \bea \label{eq:sparse_seqs}
    \abf_i =&\ {\bf 0}^{(\lambda_i)}\,\vert\,\ebf_i \otimes \gbf_i\,\vert\,{\bf 0}^{(n-2\lambda_{i+1})}\,\vert\,\fbf_i \otimes \hbf_i\,\vert\,{\bf 0}^{(\lambda_i)},\\
    \bbf_i =&\ {\bf 0}^{(\frac{n}{2}-\lambda_{i+1})}\,\vert\,\underline{\ebf_i^*} \otimes \hbf_i\,\vert\,{\bf 0}^{(2\lambda_{i})}\,\vert\,\fbf_i^* \otimes \gbf_i\,\vert\,{\bf 0}^{(\frac{n}{2}-\lambda_{i+1})}.
    \eea
    Then $\abf_1, \bbf_1, \cdots, \abf_{k+2d}, \bbf_{k+2d}$ are quasi-symmetric and supplementary sequences of length $n$ over $\{0\} \cup S_{\Cnum}$, and 
    \ben
    \sum_{i=1}^{k+2d} C_{a_i} (\tau) +  C_{b_i} (\tau) = \begin{cases}
        n,& \tau = 0\\
        0,& \tau \neq 0
    \end{cases}.
    \een
\end{theorem}
\begin{proof}
    It is obvious from \eqref{eq:sparse_seqs} that $\abf_1, \bbf_1$, $\cdots$, $\abf_{k+2d}, \bbf_{k+2d}$ are quasi-symmetric and supplementary.
    
    For $i=1, 2, \cdots, k+2d$, let 
    \ben
    \bbf_i^{\prime} =\ {\bf 0}^{(\lambda_i)}\,\vert\,\fbf_i^* \otimes \gbf_i\,\vert\,{\bf 0}^{(n-2\lambda_{i+1})}\,\vert\,\underline{\ebf_i^*} \otimes \hbf_i\,\vert\,{\bf 0}^{(\lambda_i)}.
    \een
    Similar to the proof of Proposition \ref{prop:pre_quad}, we have
    \ben \label{eq:inser_corr}
    R_{a_i}(\tau) + R_{b_i^{\prime}}(\tau) = \begin{cases} 
        4l_i m_i,& \tau = 0\\
        0,& \tau \neq 0
    \end{cases},\ 1\leq i\leq k,
    \een
    \bea \label{eq:inser_quad_corr}
    &R_{a_{k+2i-1}}(\tau) + R_{b_{k+2i-1}^{\prime}}(\tau) + R_{a_{k+2i}}(\tau) + R_{b_{k+2i}^{\prime}}(\tau) \\
    &= \begin{cases} 
        4\left(s_{2i-1} + s_{2i}\right) t_{i}, & \tau = 0\\
        0,& \tau \neq 0
    \end{cases}, \ 1\leq i\leq d,
    \eea
    Combining \eqref{eq:ac_and_pc}, \eqref{eq:total_len}, \eqref{eq:inser_corr} and \eqref{eq:inser_quad_corr}, we have 
    \ben
    \sum_{i=1}^{k+2d}C_{a_i}(\tau) + C_{b_i^{\prime}}(\tau) = \begin{cases} 
        n,& \tau = 0\\
        0,& \tau \neq 0
    \end{cases}.
    \een
    Because $\bbf_i$ is a cyclic shift of $\bbf_i^{\prime}$, $C_{b_i}(\tau) = C_{b_i^{\prime}}(\tau)$. Therefore, 
    \ben
    \sum_{i=1}^{k+2d} C_{a_i} (\tau) +  C_{b_i} (\tau) = \begin{cases}
        n,& \tau = 0\\
        0,& \tau \neq 0
    \end{cases}.
    \een
\end{proof}

An example of Theorem \ref{thm:disjoint_symmetric}: let $K=d=1$, $\{\ebf_1, \fbf_1\}$ and $\{\gbf_i, \hbf_i\}, i=1, 2, 3$, are $4$-phase GCS pairs of length $3$ [c.f. \eqref{gcs:3}], and $\{\ebf_2, \fbf_2, \ebf_3, \fbf_3\} \in CBS(8, 7)$ [c.f. \eqref{eq:cbs87}], then $n=4\times 3\times 3 + 4\times(8+7)\times 3 = 216$, $\lambda_1 = 0, \lambda_2 = 9, \lambda_3 = 33, \lambda_4 = 54$, and 
\bea
\abf_1 =&\ \ebf_1 \otimes \gbf_1\,\vert\,{\bf 0}^{(198)}\,\vert\,\fbf_1 \otimes \hbf_1,\\
\abf_2 =&\ {\bf 0}^{(9)}\,\vert\,\ebf_2 \otimes \gbf_2\,\vert\,{\bf 0}^{(150)}\,\vert\,\fbf_2 \otimes \hbf_2\,\vert\,{\bf 0}^{(9)},\\
\abf_3 =&\ {\bf 0}^{(33)}\,\vert\,\ebf_3 \otimes \gbf_3\,\vert\,{\bf 0}^{(108)}\,\vert\,\fbf_3 \otimes \hbf_3\,\vert\,{\bf 0}^{(33)},\\
\bbf_1 =&\ {\bf 0}^{(99)}\,\vert\,\underline{\ebf_1^*} \otimes \hbf_1\,\vert\,\fbf_1^* \otimes \gbf_1\,\vert\,{\bf 0}^{(99)},\\
\bbf_2 =&\ {\bf 0}^{(75)}\,\vert\,\underline{\ebf_2^*} \otimes \hbf_2\,\vert\,{\bf 0}^{(18)}\,\vert\,\fbf_2^* \otimes \gbf_2\,\vert\,{\bf 0}^{(75)}, \\
\bbf_3 =&\ {\bf 0}^{(54)}\,\vert\,\underline{\ebf_3^*} \otimes \hbf_3\,\vert\,{\bf 0}^{(66)}\,\vert\,\fbf_3^* \otimes \gbf_3\,\vert\,{\bf 0}^{(54)},
\eea
which are quasi-symmetric and supplementary sequences over $\{0\} \cup S_{\Cnum}$, and 
\ben
\sum_{i=1}^{3} C_{a_i} (\tau) +  C_{b_i} (\tau) = \begin{cases}
    216,& \tau = 0\\
    0,& \tau \neq 0
\end{cases}.
\een

\begin{remark} \label{remark:disjoint_symmetric}
    \cite[Section 5]{craigen1995signed} also constructed some periodic complementary sequences, adding up the lengths of a series of aperiodic complementary sequences. Our construction differs from that in the following two aspects.
    \begin{enumerate}
        \item As mentioned in Remark \ref{remark:corr_sum}, the definitions of "quasi-symmetric" are different. Consequently, here we do not need the two trivial sequences composed of one $1$ and multiple zeros in \cite{craigen1995signed}.
        \item Instead of directly adding up the lengths of the GCS pairs and the CBS, we first multiply the lengths with $4$-phase Golay numbers, and then add up the products. By Lemma \ref{coro:product_golay}, the total length $n$ in \eqref{eq:total_len} would be more flexible.
    \end{enumerate}
\end{remark}

\subsection{Asymptotic Existence of Hadamard Matrices} \label{subsec:main}
Now we are ready to collect the results in Section \ref{subsec:bridge}, Section \ref{subsec:seq_combine} and Section \ref{subsec:input} to prove Theorem \ref{thm:main}. 

First, we sequentially feed into Theorem \ref{thm:seq_combine} the sequences constructed by Lemma \ref{thm:disjoint_symmetric} for $2k+4d-1$ times:
\begin{enumerate}
    \item In the first iteration, we can feed $\abf_1, \bbf_1$ into Theorem \ref{thm:seq_combine} to construct a sequence $\cbf$ over $\{0\} \cup SP_{4}$, because the conditions of \eqref{cond:normal} and \eqref{cond:commu} are satisfied automatically for complex sequences;
    \item In the $t$-th iteration of applying Theorem \ref{thm:seq_combine}, the sequence $\cbf$ over $\{0\} \cup SP_{2^{t}}$ constructed in the last iteration can be regarded as $\abf$, and one of the unused sequences constructed by Lemma \ref{thm:disjoint_symmetric} can be regarded as $\bbf$. The reasons are as follows. The condition \eqref{cond:normal} holds because $R_c(\tau) = R_{c^*}(\tau)$ in the last iteration and a complex sequence always satisfies \eqref{cond:normal}. The condition \eqref{cond:commu} is also satisfied: $a_i \in \{0\} \cup SP_{2^{t}}$ can be represented by a complex matrix of order $2^{t-1}$ and $b_j \in \{0\}\cup S_{\mathcal{C}}$ can be represented by a scalar complex matrix of order $2^{t-1}$ by embedding, and their multiplication is commutative.
\end{enumerate}


Now we have obtained a sequence $\cbf$ over $SP_{2^{2k+4d}}$ with non-zero entries since $\abf_1, \bbf_1, \cdots, \abf_{k+2d}, \bbf_{k+2d}$ are supplementary, and 
\bea
C_c(\tau) =&\ R_c(\tau) + R_c(\tau-n)\\ 
=&\ \sum_{i=1}^{k+2d} R_{a_i} (\tau) +  R_{b_i} (\tau) + R_{a_i}(\tau-n) + R_{b_i}(\tau-n)\\
=&\ \sum_{i=1}^{k+2d} C_{a_i} (\tau) +  C_{b_i} (\tau) = \begin{cases}
        n,& \tau = 0\\
        0,& \tau \neq 0
    \end{cases}.
\eea
Hence $\cbf$ is a perfect sequence over  $SP_{2^{2k+4d}}$ of length $n$ [c.f. \ref{eq:total_len}]. Because there exist Sylvester-type Hadamard matrices of order $\Hbf_{2^{t}}, \forall\, t \in \Znum^{+}$ \cite{seberry2020hadamard}, the perfect sequence $\cbf$ can be used in Lemma \ref{lem:sg_hadamard} to construct a Hadamard matrix, leading to the following corollary.
\begin{Corollary}  \label{coro:exist_hadamard}
     There exist block-circulant Hadamard matrices of order $2^{2k+4d+2}\left(\sum_{i=1}^{k} l_i m_i + \sum_{i=1}^{d} \left(s_{2i-1} + s_{2i}\right)t_i\right)$ with block size $2^{2k+4d}$, where $l_1, m_1, \cdots, l_k, m_k$, $t_1, \cdots, t_d$ are $4$-phase Golay numbers and $s_{2i-1},  s_{2i}$ are the lengths of a CBS for $\forall\, 1\leq i \leq d$.  
\end{Corollary}

\begin{remark}
    Owing to the differences mentioned in Remark \ref{remark:corr_sum} and Remark \ref{remark:disjoint_symmetric}, the results in Corollary \ref{coro:exist_hadamard} is different from those in \cite{craigen1995signed}\cite[Theorem 3.11]{livinskyi2012asymptotic}, where the order is $2^{2k+4d+3}\left(1+2\sum_{i=1}^{k} l_i + 2\sum_{i=1}^{2d} s_{i}\right)$ and the block size is $2^{2k+4d+2}$.
\end{remark}

Our Matlab\textsuperscript{TM} codes that can generate Hadamard matrices from $4$-phase GCS pairs are available online: \url{https://github.com/csrlab-fudan/gcs_hadamard}.

\begin{table*}[t]
\caption{Results obtained from Sequences in Theorem \ref{thm:disjoint_symmetric}}  \label{tab:bnk}
    \centering
    \begin{tabular}{|c|c|c|c|c|c|}
        \hline
            $\gamma$ & $b_{\gamma}^{(0)}$ & $b_{\gamma}^{(1)}$ & $b_{\gamma}^{(2)}$ &$b_{\gamma}^{(3)}$ &$b_{\gamma}^{(4)}$ \\
        \hline
            $4$ & $546$ & $1030$ & $1030$ &$1030$ &$1030$\\ 
        \hline
            $6$ & $\numprint{436146}$ & $\numprint{161926498}$ & $\numprint{11736430180}$ & $\numprint{313523649186}$ & $\geq 2^{40}$\\
        \hline
            $8$ & $\geq 2^{44}$ & $\geq 2^{43}$ & $\geq 2^{42}$ & $\geq 2^{41}$ & $\geq 2^{40}$\\
        \hline
    \end{tabular}
\end{table*}

Finally we give a crude lower bound of $\gamma(N) \triangleq 2k+4d$, for which $\sum_{i=1}^{k} l_i m_i + \sum_{i=1}^{d} \left(s_{2i-1} + s_{2i}\right)t_i$ can cover $\forall\, N \in \Znum^+$.

Given $\kappa \in \undertilde{\mathcal{S}}$, suppose $\sum_{i=1}^{k} l_i m_i + \sum_{i=1}^{d} \left(s_{2i-1} + s_{2i}\right)t_i$ covers all the lengths of form $\kappa^{i} b$ for any $1\leq b \leq b_{\gamma}^{(i)}$ with given integers $\gamma$ and $i$. Let $b_{\gamma_0}^{(i_0)} \triangleq \max\{b_{\gamma}^{(i)}\vert\gamma=4, 6, 0\leq i \leq i_{max}\}$, 
$\xi \triangleq \lfloor \log_{\kappa} \min\{b_{\gamma_0}^{(i_0)}, b_{8}^{(0)}\}\rfloor$ and $P \triangleq \kappa^{\xi}$ ($i_{max}$ is a proper integer such that $i_0 \leq \xi$). Expand $N$ to the base $P$ as $N = \sum_{i=0}^{q} N_iP^i$, where $q = \lfloor \log_{P}N\rfloor$ and $0\leq N_i< P$. Define $\mathcal{I} \triangleq \{i \,\vert\,N_i\neq0, i=1, \cdots, q\}$ and $r \triangleq \abs{\mathcal{I}}$. For given $i \in \mathcal{I}$, we have $N_i P^i = N_i \kappa^{i_0} \kappa^{i\xi-i_0}$. Because $N_i <P \leq b_{\gamma_0}^{(i_0)} $, by assumption $N_i \kappa^{i_0}$ can be covered if $\gamma = \gamma_0$. Besides, the lengths of $\gbf_1, \hbf_1, \cdots, \gbf_{k+2d}, \hbf_{k+2d}$ in Theorem \ref{thm:disjoint_symmetric} can be enlarged by a factor of $\kappa^{i\xi-i_0}$ by Proposition \ref{prop:quater_seq}. Hence $N_i P^i$  for given $i \in \mathcal{I}$ can be covered if $\gamma = \gamma_0$. For $N_0 \neq 0$, because $N_0 < P \leq b_8^{(0)}$, $N_0$ can be covered if $\gamma = 8$. Hence $\forall\, N\in \Znum^+$ can be covered if 
\ben
\gamma(N) = \begin{cases}
    \gamma_0 r, & N_0 = 0 \\
    8+\gamma_0r, & N_0 \neq 0
\end{cases}.
\een

For example, let $N_{max}=2^{44}$, $i_{max} = 4$, $\kappa = 2$, $b_{\gamma}^{(i)}$ obtained by computational verification are listed in Table \ref{tab:bnk}, where $b_8^{(0)} \geq 2^{44}$ and $b_{\gamma_0}^{(i_0)} = b_6^{(4)} \geq 2^{40}$. Then we have $\xi \geq 40 > i_0$, $P\geq 2^{40}$, and $r\leq q \leq \lfloor \frac{1}{40}\log_{2}N\rfloor$. Hence $\forall\, N \in \Znum^+$ can be covered if $\gamma(N) =  8+6\lfloor \frac{1}{40}\log_{2}N\rfloor$. The C\texttt{++} codes for the above verification are also available online: \url{https://github.com/csrlab-fudan/gcs_hadamard}. 

Combining the above lower bound and Corollary \ref{coro:exist_hadamard} leads to the conclusion in Theorem \ref{thm:main}: there exist block-circulant Hadamard matrices of order $2^t m$ with block size $2^{t-2}$ for any odd number $m$, where $t= 6\lfloor \frac{1}{40}\log_{2}m\rfloor + 10$.

\section{Conclusions} \label{SEC:con}
In this paper, we construct $4$-phase Golay complementary sequence (GCS) set of cardinality $2^{3+\lceil \log_2 r \rceil}$ with arbitrary sequence length $n$, where the $10^{13}$-base expansion of $n$ has $r$ nonzero digits. We also obtain an improved asymptotic existence of Hadamard matrices: there exist Hadamard matrices of order $2^t m$ for any odd number $m$, where $t = 6\lfloor \frac{1}{40}\log_{2}m\rfloor + 10$.

A promising way to further improve the asymptotic existence of Hadamard matrices is to find a construction of Hadamard matrices of order $2^{t}n$ from a $4$-phase GCS set of cardinality $2^{t-1}$ with arbitrary length $n$, which can be seen as a generalization of Proposition \ref{prop:gs}. If found, combined with the results in Section \ref{SEC:gcs}, $t$ only needs to increase in a log log rate, a much slower rate than a logarithmic one that we currently have.
\bibliographystyle{IEEEtran}
\bibliography{gcs_hadamard}


\end{document}